%% file: BrDrSaWo_Revised.tex
\documentclass{svmult}

\input{macros}

\usepackage[usenames,dvipsnames]{color}
\newcommand{\DEF}{{\,:=\,}}
\newcommand{\PT}[1]{\mathbf{#1}}
\DeclareMathOperator{\bal}{Bal}
\DeclareMathOperator{\betafcn}{B}
\DeclareMathOperator{\dd}{\D}
\DeclareMathOperator{\digammafcn}{\psi}
\DeclareMathOperator{\gammafcn}{\Gamma}
\DeclareMathOperator{\IncompleteBetaRegularized}{I}
\DeclareMathOperator{\supp}{supp}
\DeclareMathOperator{\HyperTildeF}{\mathop{\mathbf{F}\/}\nolimits\!}

\newcommand{\HypergeomReg}[5]{{\sideset{_#1}{_#2}
\HyperTildeF\!\left(\substack{\displaystyle#3\\\displaystyle#4};#5\right)}}

\begin{document}

\title*{Logarithmic and Riesz Equilibrium for Multiple Sources on the Sphere --- the Exceptional Case}
\titlerunning{Logarithmic and Riesz Equilibrium for Multiple Sources on the Sphere}

\author{Johann S.\ Brauchart \and Peter D.\ Dragnev \and Edward B.\ Saff \and Robert S.\ Womersley}

\authorrunning{J.~S.~Brauchart, P.~D.~Dragnev, E.~B.~Saff, and R.~S.~Womersley}

\institute{
 Johann S. Brauchart 
 \at Institute of Analysis and Number Theory, Graz University of Technology, Kopernikusgasse 24/II, 8010 Graz, Austria \\
 \email{j.brauchart@tugraz.at}
 \and
 Peter D.\ Dragnev
 \at Department of Mathematical Sciences, Indiana University - Purdue University, Fort Wayne, \\ IN 46805, USA \\
 \email{dragnevp@ipfw.edu}
 \and
 Edward B.\ Saff (\Letter)
 \at Center for Constructive Approximation, Department of Mathematics, Vanderbilt University, Nashville, TN 37240, USA \\
 \email{edward.b.saff@vanderbilt.edu}
 \and
 Robert S.\ Womersley
 \at School of Mathematics and Statistics, University of New South Wales, Sydney, NSW, 2052, \\ Australia \\
 \email{r.womersley@unsw.edu.au}
}

\maketitle

\index{Brauchart, Johann S.}
\index{Dragnev, Peter D.}
\index{Saff, Edward B.}
\index{Womersley, Robert S.}

{\center{\it To Ian Sloan,  an outstanding mathematician, mentor, and colleague, \\ with much appreciation for his insights, guidance, and friendship.}}

\vskip 5mm
\abstract{
We consider the minimal discrete and continuous energy problems on the unit sphere $\mathbb{S}^d$ in the Euclidean space $\mathbb{R}^{d+1}$ in the presence of an external field due to finitely many localized charge distributions on $\mathbb{S}^d$, where the energy arises from the Riesz potential $1/r^s$ ($r$ is the Euclidean distance) for the critical Riesz parameter $s = d - 2$ if $d \geq 3$ and the logarithmic potential $\log(1/r)$ if $d = 2$. Individually, a localized charge distribution is either a point charge or assumed to be rotationally symmetric. The extremal measure solving the continuous external field problem for weak fields is shown to be the uniform measure on the sphere but restricted to the exterior of spherical caps surrounding the localized charge distributions. The radii are determined by the relative strengths of the generating charges. Furthermore, we show that the minimal energy points solving the related discrete external field problem are confined to this support. For $d-2\leq s<d$, we show that for point sources on the sphere, the equilibrium measure has support in the complement of the union of specified spherical caps about the sources. Numerical examples are provided to illustrate our results.
}

\section{Introduction} 
\label{Intr}

Let $\mathbb{S}^d \DEF \{ \PT{x} \in \mathbb{R}^{d+1} : |\PT{x}|=1 \}$ be the unit sphere in $\mathbb{R}^{d+1}$, where $|\PT{\cdot}|$ denotes the Euclidean norm. Given a compact set $E\subset \mathbb{S}^d$, consider the class $\mathcal{M}(E)$ of unit positive Borel measures supported on $E$. For $0<s<d$ the {\it Riesz $s$-potential} and {\it Riesz $s$-energy} of a measure $\mu \in \mathcal{M}(E)$ are given, respectively, by
\begin{equation*}
U_s^\mu(\PT{x}) \DEF \int k_s( \PT{x}, \PT{y} ) \dd \mu(\PT{y}), \ \ \PT{x} \in \mathbb{R}^{d+1}, \quad \mathcal{I}_s(\mu) \DEF \int \int k_s( \PT{x}, \PT{y} ) \dd \mu(\PT{x}) \dd \mu(\PT{y}),
\end{equation*}
where $k_s (\PT{x}, \PT{y}) \DEF |\PT{x}-\PT{y}|^{-s}$ for $s>0$ is the so-called {\it Riesz kernel}. For the case $s=0$ we use the logarithmic kernel $k_0(\PT{x}, \PT{y}) \DEF \log(1/|\PT{x}-\PT{y}|)$. The {\it $s$-capacity} of $E$ is then defined as $C_s(E) \DEF 1/ W_s(E)$ for $s>0$ and $C_0 (E)=\exp(-W_0 (E))$, where $W_s(E) \DEF \inf \{ \mathcal{I}_s(\mu) : \mu \in \mathcal{M}(E) \}$. A property is said to hold {\em quasi-everywhere (q.e.)} if the exceptional set has $s$-capacity zero. When $C_s(E)>0$, there exists a unique minimizer $\mu_E = \mu_{s,E}$, called the {\it $s$-equilibrium measure on $E$}, such that $\mathcal{I}_s(\mu_E) = W_s(E)$. The $s$-equilibrium measure 
is just the normalized surface area measure on $\mathbb{S}^d$ which we denote with $\sigma_d$. For more details see \cite[Chapter II]{Landkof1972}.

We remind the reader that the $s$-energy of $\mathbb{S}^d$ is given by
\begin{equation} \label{SphereEnergy} 
U_s^{\sigma_d}(\PT{x}) = \mathcal{I}_s(\sigma_d) = W_s(\mathbb{S}^d) = \frac{\gammafcn(d)\gammafcn((d-s)/2)}{2^s\gammafcn(d/2)\gammafcn(d-s/2)}, \qquad 0 < s < d,
\end{equation}
and the logarithmic energy of $\mathbb{S}^d$ is given by
\begin{equation*} 
\begin{split}
U_0^{\sigma_d}(\PT{x}) = \mathcal{I}_0(\sigma_d) = W_0(\mathbb{S}^d) 
&= \frac{\D W_s(\mathbb{S}^d)}{\D s}\Big|_{s=0} = - \log( 2 ) + \frac{1}{2} \left( \digammafcn( d ) - \digammafcn( d/2 ) \right),
\end{split}
\end{equation*}
where $\digammafcn( s ) \DEF \gammafcn^\prime( s ) / \gammafcn( s )$ is the digamma function.
%
Using cylindrical coordinates
\begin{equation}
\PT{x} = ( \sqrt{1-u^2} \; \overline{\PT{x}}, u ), \qquad -1 \leq u \leq 1, \;  \overline{\PT{x}} \in \mathbb{S}^{d-1}, \label{Cylindrical}
\end{equation}
we can write the decomposition
\begin{equation} \label{Measure.Decomposition}
\dd \sigma_{d}(\PT{x}) = \frac{\omega_{d-1}}{\omega_{d}} \left( 1 - u^{2} \right)^{d/2-1} \, \dd u \, \dd \sigma_{d-1}(\overline{\PT{x}}).
\end{equation}
Here $\omega_{d}$ is the surface area of $\mathbb{S}^{d}$ and the ratio of these areas can be evaluated as
\begin{equation}
\omega_0 = 2, \quad \frac{\omega_{d}}{\omega_{d-1}} = \int_{-1}^{1} \left( 1 - u^2 \right)^{d/2-1} \dd u =
\frac{\sqrt{\pi}\gammafcn(d/2)}{\gammafcn((d+1)/2)} = 2^{d-1} \frac{\left[\gammafcn(d/2)\right]^{2}}{\gammafcn(d)}. \label{omega.ratio}
\end{equation}

We shall refer to a non-negative lower semi-continuous function $Q:\mathbb{S}^d \to [0,\infty]$ such that $Q(\PT{x})<\infty$ on a set of positive Lebesgue surface area measure as an {\em external field}. The weighted energy associated with $Q$ is then given by
\begin{equation} \label{energy}
I_Q(\mu) \DEF \mathcal{I}_s(\mu) + 2 \int Q(\PT{x}) \dd \mu(\PT{x}).
\end{equation}

\begin{definition}
The minimal energy problem on the sphere in the presence of the external field $Q$ refers to the quantity
\begin{equation} \label{Vq}
V_Q \DEF \inf \left\{ I_Q(\mu) : \mu \in \mathcal{M}(\mathbb{S}^d) \right\}.
\end{equation}
A measure $\mu_Q=\mu_{Q,s} \in \mathcal{M}(\mathbb{S}^d)$ such that $I_Q(\mu_Q) = V_Q$ is called an {\em $s$-extremal} (or {\em $s$-equilibrium}) {\em measure associated with $Q$}.
\end{definition}

The discretized version of the minimal $s$-energy problem is also of interest. The associated optimal point configurations have a variety of possible 
applications, such as for generating radial basis functions on the sphere that are used in the numerical solutions to PDEs (see, e.g., \cite{LeGSloWen2012a}, \cite{LeGSloWen2012b}). 

Given a positive integer $N$, we consider the optimization problem
\begin{equation} \label{Eq}
\mathcal{E}_{Q,N}\DEF \min_{\{ \PT{x}_1, \ldots, \PT{x}_N \}\subset \mathbb{S}^d} \sum_{1\leq i\not= j \leq N} \big[ k_s(\PT{x}_i,\PT{x}_j)+Q(\PT{x}_i)+Q(\PT{x}_j) \big].
\end{equation}
A system that minimizes the discrete energy is called an {\em optimal (minimal) $s$-energy $N$-point configuration w.r.t. $Q$}. The field-free case $Q\equiv 0$ is particularly important.
 
The following Frostman-type result as stated in \cite{DraSaf2007} summarizes the existence and uniqueness properties for $s$-equilibrium measures on $\mathbb{S}^d$ in the presence of external fields (see also \cite[Theorem I.1.3]{SafTot1997} for the complex plane case and \cite{Zorii2003} for more general spaces).

\begin{proposition} \label{prop:1}
Let $0\leq s<d$. For the minimal $s$-energy problem on $\mathbb{S}^d$ with external field $Q$ the following properties hold:
\begin{enumerate}
\item[\rm (a)] $V_Q$ is finite.
\item[\rm (b)] There exists a unique $s$-equilibrium measure $\mu_Q =\mu_{Q,s}\in \mathcal{M}(\mathbb{S}^d)$ associated with~$Q$.
Moreover, the support $S_Q$ of this measure is contained in the compact set $E_M \DEF \{ \PT{x} \in \mathbb{S}^d : Q(\PT{x}) \leq M \}$ for some $M>0$.
\item[\rm (c)] The measure $\mu_Q$ satisfies the variational inequalities
\begin{align}
U_s^{\mu_Q}(\PT{x}) + Q(\PT{x}) &\geq F_Q \quad \text{q.e. on $\mathbb{S}^d$,} \label{VarEq1} \\
U_s^{\mu_Q}(\PT{x}) + Q(\PT{x}) &\leq F_Q \quad \text{for all $\PT{x} \in S_Q$,} \label{VarEq2}
\end{align}
where
\begin{equation}
F_Q \DEF V_Q - \int Q(\PT{x}) \dd \mu_Q(\PT{x}). \label{VarConst}
\end{equation}
\item[\rm (d)] Inequalities \eqref{VarEq1} and \eqref{VarEq2} completely characterize the extremal measure $\mu_Q$ in the sense that if $\nu \in \mathcal{M}(\mathbb{S}^d)$ is a measure with finite $s$-energy such that 
\begin{align}
U_s^{\nu}(\PT{x}) + Q(\PT{x}) &\geq C \quad \text{q.e. on $\mathbb{S}^d$,} \label{VarEq3} \\
U_s^{\nu}(\PT{x}) + Q(\PT{x}) &\leq C \quad \text{for all $\PT{x} \in \supp( \nu )$} \label{VarEq4}
\end{align}
for some constant $C$, we have then $\mu_Q = \nu$ and $F_Q = C$.
\end{enumerate}
\end{proposition}

\begin{remark} We note that a similar statement holds true when $\mathbb{S}^d$ is replaced by any compact subset $K\subset \mathbb{S}^d$ of positive $s$-capacity. 
\end{remark}

The explicit determination of $s$-equilibrium measures or their support is not an easy task. In \cite{DraSaf2007} an external field exerted by a single point mass on the sphere was applied to establish that, in the field-free case, minimal $s$-energy $N$-point systems on~$\mathbb{S}^d$, as defined in \eqref{Eq}, are ``well-separated'' for $d-2<s<d$. Axis-supported external fields were studied in \cite{BraDraSaf2009} and rotationally invariant external fields on $\mathbb{S}^2$ in \cite{Bil2016}. The separation of minimal $s$-energy $N$-point configurations for more general external fields, namely Riesz $s$-potentials of signed measures with negative charge outside the unit sphere, was established in \cite{BraDraSaf2014}.

Here we shall focus primarily on the exceptional case when $s=d-2$ and $Q$ is the external field exerted by finitely many localized charge distributions. Let $\PT{a}_1,\PT{a}_2,\dots, \PT{a}_m \in \mathbb{S}^d$ be $m$ fixed points with associated positive charges $q_1,q_2,\dots, q_m$. Then the external field is given by
\begin{equation} \label{DiscrField}
Q(\PT{x}):=\sum_{i=1}^m q_i \, k_{d-2}(\PT{a}_i , \PT{x}).
\end{equation} 
For sufficiently small charges $q_1, \dots, q_m$ we completely characterize the ${(d-2)}$-equilibrium measure for the external field~\eqref{DiscrField}. 

The outline of the paper is as follows. In Section~\ref{Pre}, we introduce some notion from potential theory utilized in our analysis. In Section~\ref{Log}, we present the important case of the unit sphere in the $3$-dimensional space and logarithmic interactions.  
An interesting corollary in its own right for discrete external fields in the complex plane is exhibited as well. 
The situation when $d \geq 3$, considered in Section~\ref{Rie}, is more involved as there is a loss of mass in the balayage process. 
Finally, in Section \ref{Int}, we derive a result on regions free of optimal points and formulate an open problem.

\section{Signed Equilibria, Mhaskar-Saff $\mathcal{F}$-Functional, and Balayage}
\label{Pre}

A significant role in our analysis is played by the so-called {\em signed equilibrium} (see~\cite{BraDraSaf2009, BraDraSaf2014}).

\begin{definition} \label{signed} Given a compact subset $E\subset \mathbb{R}^p$, $p\geq 3$, and an external field $Q$, we call a signed measure $\eta_{E,Q}=\eta_{E,Q,s}$ supported on $E$ and of total charge $\eta_{E,Q}(E)=1$ {\em a signed $s$-equilibrium on $E$ associated with $Q$} if its weighted Riesz $s$-potential is constant on $E$:
\begin{equation}
U_s^{\eta_{E,Q}}(\PT{x}) + Q(\PT{x}) = F_{E,Q} \qquad \text{for all $\PT{x} \in E$.} \label{signedeq}
\end{equation}
\end{definition}

We note that if the signed equilibrium exists, it is unique (see \cite[Lemma 23]{BraDraSaf2009}). In view of \eqref{VarEq1} and \eqref{VarEq2}, the signed equilibrium on $S_Q$ is actually a non-negative measure and coincides with the $s$-extremal measure associated with $Q$, and hence can be obtained by solving a singular integral equation on $S_Q$. Moreover, for the equilibrium support we have that $S_Q\subset {\rm supp}(\eta_{E,Q}^+)$ whenever $S_Q\subset E\subset \mathbb{S}^d$ (see \cite[Theorem 9]{BraDraSaf2014}).

An important tool in our analysis is the Riesz analog of the \emph{Mhaskar-Saff \linebreak {$F$-functional}} from classical logarithmic potential theory in the plane (see \cite{MhaSaf1985} and \cite[Chapter IV, p. 194]{SafTot1997}).

\begin{definition} 
The {\em $\mathcal{F}_s$-functional} of a compact subset $K\subset \mathbb{S}^d$ of
positive $s$-capacity is defined as 
\begin{equation} \label{Functional}
\mathcal{F}_s(K) \DEF W_s(K) + \int Q(\PT{x}) \, \dd \mu_K(\PT{x}),
\end{equation}
where $W_s(K)$ is the $s$-energy of $K$ and $\mu_K$ is the $s$-equilibrium measure on $K$.
\end{definition}

\begin{remark} \label{FuncSignedEqRel}
As pointed out in \cite{BraDraSaf2009, BraDraSaf2014}, when $d-2\leq s<d$, a relationship exists between the signed $s$-equilibrium constant in \eqref{signedeq} and the $\mathcal{F}_s$-functional \eqref{Functional}, namely $\mathcal{F}_s (K)=F_{K,Q}$. Moreover, the equilibrium support minimizes the $\mathcal{F}$-functional; i.e., if $d-2\leq s<d$ and $Q$ is an external field on $\mathbb{S}^d$, then the $\mathcal{F}_s$-functional is minimized for
$S_Q = \supp (\mu_Q)$ (see \cite[Theorem 9]{BraDraSaf2009}).
\end{remark}

A tool we use extensively is the {\em Riesz $s$-balayage measure} (see \cite[Section~4.5]{Landkof1972}). Given a measure $\nu$ supported on $\mathbb{S}^d$ and a compact subset $K\subset \mathbb{S}^d$, the measure $\widehat{\nu}:=\bal_s (\nu,K)$ is called the {\em Riesz $s$-balayage} of $\nu$ onto $K$, $d-2\leq s <d$, if $\widehat{\nu}$ is supported on $K$ and 
\begin{equation}\label{sBal}
\begin{split}
U_s^{\widehat{\nu}}(\PT{x}) &= U_s^{\nu}(\PT{x}) \qquad \text{on $K$,} \\
U_s^{\widehat{\nu}}(\PT{x}) &\leq U_s^{\nu}(\PT{x}) \qquad \text{on $\mathbb{S}^d$.}
\end{split}
\end{equation}
In general, there is some loss of mass, namely $\widehat{\nu}(\mathbb{S}^d)< \nu(\mathbb{S}^d)$. However, in the logarithmic interaction case $s=0$ and $d=2$, the mass of the balayage measures is preserved, but as in the classical complex plane potential theory we have equality of potentials up to a constant term
\begin{equation}\label{LogBal}
\begin{split}
U_0^{\widehat{\nu}}(\PT{x}) &= U_0^{\nu}(\PT{x}) + C \qquad \text{on $K$,}\\
U_0^{\widehat{\nu}}(\PT{x}) &\leq U_0^{\nu}(\PT{x}) + C \qquad \text{on $\mathbb{S}^2$.}
\end{split}
\end{equation}
Balayage of a signed measure $\eta$ is achieved by taking separately the balayage of its positive and its negative part in the Jordan decomposition $\eta=\eta^+ - \eta^-$. An important property is that we can take balayage in steps: if $F\subset K \subset \mathbb{S}^d$, then
\begin{equation} \label{eq:balayage.in.steps}
\bal_s( \nu, F) = \bal_s( \bal_s(\nu,K), F ).
\end{equation}
We also use the well-known relation
\begin{equation} \label{eq:balayage.partition}
\bal_s( \nu, K) = \nu_{\vert_{K}} + \bal_s( \nu_{\vert_{\mathbb{S}^d \setminus K}}, K). 
\end{equation}

\section{Logarithmic Interactions on $\mathbb{S}^2$} 
\label{Log}

We first state and prove our main theorem for the case of logarithmic interactions on $\mathbb{S}^2$.
We associate with $Q$ (or equivalently with $\{ \PT{a}_i \}$ and $\{ q_i \}$) the total charge 
\begin{equation*} 
q := q_1+ \cdots + q_m
\end{equation*}
the vector 
\begin{equation}\label{epsilon}
\boldsymbol{\epsilon}=(\epsilon_1,\dots, \epsilon_m), \qquad \epsilon_i := 2 \, \sqrt{ \frac{ q_i }{1 + q}}, \quad i = 1, \dots, m,
\end{equation}
and the set  
\begin{equation} \label{Sigma_epsilon}
\Sigma_{\boldsymbol{\epsilon}}=\bigcap_{i=1}^m \Sigma_{i,\epsilon_i}, \qquad \Sigma_{i,\epsilon} := \big\{ \PT{x} \in \mathbb{S}^2 : | \PT{x}-\PT{a}_i | \geq \epsilon \big\}, \quad i = 1, \dots, m, \, \epsilon \geq 0.
\end{equation}
More generally, with any vector $\boldsymbol{\gamma}=(\gamma_1,\dots, \gamma_m)$ with non-negative components we associate the set $\Sigma_{\boldsymbol{\gamma}}=\bigcap_{i=1}^m \Sigma_{i,\gamma_i}$. Note: if $\boldsymbol{\gamma} \leq \boldsymbol{\epsilon}$ (i.e., $\gamma_i\leq \epsilon_i$, $1 \leq i \leq m$), then $\Sigma_{\boldsymbol{\epsilon}}\subset \Sigma_{\boldsymbol{\gamma}}$.

\begin{theorem} \label{thm:main}
Let $d = 2$ and $s = 0$. Let $Q$, $\boldsymbol{\epsilon}$, and $\Sigma_{\boldsymbol{\epsilon}}$ be defined by \eqref{DiscrField}, \eqref{epsilon}, and \eqref{Sigma_epsilon}. Suppose that $\Sigma_{i,\epsilon_i}^c \cap \Sigma_{j,\epsilon_j}^c = \emptyset$, $1 \leq i < j \leq m$ ($K^c$ denotes the  complement of $K$ relative to the sphere). Then the logarithmic extremal measure associated with $Q$ is $\mu_Q = \left( 1 + q \right) {\sigma_2}_{\vert_{\Sigma_{\boldsymbol{\epsilon}}}}$ and the extremal support is $S_Q = \Sigma_{\boldsymbol{\epsilon}}$.
\end{theorem}

\begin{remark} \label{rmk:main} The theorem has the following electrostatics interpretation. As positively charged particles $\PT{a}_i$ are introduced on a positively pre-charged unit sphere, they 
create 
charge-free regions which we call {\em regions of electrostatic influence}.  The theorem then states that if the potential interaction is logarithmic and the charges of the particles are sufficiently small (so that the regions of influence do not overlap), then these regions are perfect spherical caps $\Sigma_{i,\epsilon_i}^c$ whose radii depend only on the amount of charge and the position of the particles. In Section \ref{Int}, we partially investigate what happens when the $q_i$'s increase beyond the critical values imposed by the non-overlapping conditions $\Sigma_{i,\epsilon_i}^c \cap \Sigma_{i,\epsilon_j}^c = \emptyset$, $1\leq i <j \leq m$.
\end{remark}

\begin{proof} 
Let $m = 1$. This case has already been solved in \cite{BraDraSaf2009}. By \cite[Theorem 17]{BraDraSaf2009}, the signed equilibrium on $\Sigma_{\gamma}$ associated with $Q( \PT{x} ) := q \log \frac{1}{| \PT{x} - \PT{a} |}$, $\PT{a} \in \mathbb{S}^2$, is given by 
\begin{align} 
\eta_{\Sigma_\gamma,Q} 
&= \left( 1 + q \right) \bal_0( \sigma_2, \Sigma_\gamma ) - q \, \bal_0( \delta_{\PT{a}}, \Sigma_\gamma ) \notag \\
&= \left( 1 + q \right) {\sigma_2}_{\vert_{\Sigma_\gamma}} + \left( 1 + q \right) \left( \frac{\gamma^2}{4} - \frac{q}{1 + q} \right) \beta, \label{eq:eta.simple}
\end{align}
where $\beta$ is the normalized Lebesgue measure on the boundary circle of $\Sigma_\gamma$. 
%

The logarithmic extremal measure on $\mathbb{S}^2$ associated with $Q$ is then given by
\begin{equation*}
\mu_Q = \left( 1 + q \right) {\sigma_2}_{\vert_{\Sigma_\epsilon}}, \qquad \text{where $\epsilon = 2 \, \sqrt{\frac{q}{1+q}}$.}
\end{equation*}

Let $\xi := \langle \PT{x} , \PT{a} \rangle$ and $\gamma^2 = 2 ( 1 - t )$, where $t$ is the projection of the boundary circle $\partial \Sigma_{\gamma}$ onto the $\PT{a}$-axis. For future reference, by \cite[Lemmas 39 and 41]{BraDraSaf2009} we have 
\begin{equation} \label{BalSigma}
U_0^{\bal_0( \sigma_2, \Sigma_\gamma )}( \PT{x} ) = 
\begin{cases}
W_0( \Sigma_\gamma ), & \PT{x} \in \Sigma_\gamma, \\[.5em]
W_0( \Sigma_\gamma ) + \dfrac{1}{2} \log \dfrac{1 + t}{1 + \xi}, & \PT{x} \in \Sigma_\gamma^c
\end{cases}
\end{equation}
and
\begin{equation} \label{BalDelta}
U_0^{\bal_0( \delta_{\PT{a}}, \Sigma_\gamma )}( \PT{x} ) = U_0^{\delta_{\PT{a}}}( \PT{x} ) +
\begin{cases}
\dfrac{1}{2} \log \dfrac{1-t}{1+t} - \dfrac{1}{2} \log \dfrac{1-t}{2}, & \PT{x} \in \Sigma_\gamma, \\[.5em]
\dfrac{1}{2} \log \dfrac{1-\xi}{1+\xi} - \dfrac{1}{2} \log \dfrac{1-t}{2}, & \PT{x} \in \Sigma_\gamma^c.
\end{cases}
\end{equation}
Moreover, $\widehat{\nu} = \bal_0( {\sigma_2}_{\vert_{\Sigma_{\gamma}^c}}, \Sigma_{\gamma})$ and $\widehat{\delta_{\PT{a}}} = \bal_0( \delta_{\PT{a}}, \Sigma_{\gamma})$ are multiples of $\beta$:
\begin{equation} \label{BalLog}
\widehat{\nu} = \sigma_2(\Sigma_{\gamma}^c) \, \beta =\frac{ \gamma^2}{4} \beta=\frac{1-t}{2}\beta, \quad \quad \widehat{\delta_{\PT{a}}}=\beta,
\end{equation}

Let $m \geq 2$. 
First, we determine the signed equilibrium on the set $\Sigma_{\boldsymbol{\gamma}}$, $\boldsymbol{\gamma} \leq \boldsymbol{\epsilon}$, associated with~$Q$. 
%
%
%
We consider the signed measure 
\begin{equation*}
\tau := \left( 1 + q \right) \bal_0( \sigma_2, \Sigma_{\boldsymbol{\gamma}} ) - \bal_0( q_1 \, \delta_{\PT{a}_1} + \cdots + q_m \, \delta_{\PT{a}_m}, \Sigma_{\boldsymbol{\gamma}} ). 
\end{equation*}
As balayage under logarithmic interaction is linear and preserves mass, we have\footnote{The mass of a signed measure $\mu$ is defined as $\| \mu \| := \int \dd \mu$.} 
\begin{equation*}
\| \tau \| = \left( 1 + q \right) \| \sigma_2 \| - \sum_{i = 1}^m q_i \, \| \delta_{\PT{a}_i} \| = 1 + q - \sum_{i = 1}^m q_i = 1. 
\end{equation*}

The hypotheses on $\Sigma_{\boldsymbol{\epsilon}}$ and $\Sigma_{\boldsymbol{\epsilon}} \subset \Sigma_{\boldsymbol{\gamma}}$, $\boldsymbol{\gamma} \leq \boldsymbol{\epsilon}$, imply the non-overlapping conditions
\begin{equation*}
\Sigma_{i, \gamma_i}^c \cap \Sigma_{j, \gamma_j}^c = \emptyset, \qquad 1 \leq i < j \leq m.
\end{equation*}
%
%
For $i = 1, \dots, m$ let 
\begin{equation} \label{BalSeq} 
\nu_i := {\sigma_2}_{\vert_{\Sigma_{i,\gamma_i}^c}} ,\qquad  \widehat{\nu_i} := \bal_0 (\nu_i, \Sigma_{i,\gamma_i}), \qquad  \widehat{\delta_{\PT{a}_i}} := \bal_0 ( \delta_{\PT{a}_i}, \Sigma_{i,\gamma_i}).
\end{equation}
Since $\Sigma_{i,\gamma_i} \supset \Sigma_{\boldsymbol{\gamma}}$, balayage in steps (cf. \eqref{eq:balayage.in.steps}) yields
\begin{equation*}
\bal_0( \nu_i, \Sigma_{\boldsymbol{\gamma}} ) = \bal_0( \bal_0( \nu_i, \Sigma_{i,\gamma_i} ), \Sigma_{\boldsymbol{\gamma}} ) = \bal_0( \nu_i, \Sigma_{i,\gamma_i} ) = \widehat{\nu_i}.
\end{equation*}
The second step follows because $\widehat{\nu_i}$ is supported on $\partial \Sigma_{i,\gamma_i}$ which is included in $\partial \Sigma_{\boldsymbol{\gamma}}$. 
Hence
\begin{equation*}
\begin{split}
\bal_0( \sigma_2, \Sigma_{\boldsymbol{\gamma}} ) 
&= \bal_0( {\sigma_2}_{\vert_{\Sigma_{\boldsymbol{\gamma}}}} + \nu_1 + \cdots + \nu_m, \Sigma_{\boldsymbol{\gamma}} ) \\
&= {\sigma_2}_{\vert_{\Sigma_{\boldsymbol{\gamma}}}} + \sum_{i=1}^m \bal_0( \nu_i, \Sigma_{\boldsymbol{\gamma}} ) \\
&= {\sigma_2}_{\vert_{\Sigma_{\boldsymbol{\gamma}}}} + \sum_{i=1}^m \widehat{\nu_i}.
\end{split}
\end{equation*}
Likewise,
\begin{equation*}
\bal_0( \delta_{\PT{a}_i}, \Sigma_{\boldsymbol{\gamma}} ) = \bal_0( \bal_0( \delta_{\PT{a}_i}, \Sigma_{i,\gamma_i} ), \Sigma_{\boldsymbol{\gamma}} ) = \bal_0( \delta_{\PT{a}_i}, \Sigma_{i,\gamma_i} ) = \widehat{\delta_{\PT{a}_i}}.
\end{equation*}
Hence, we obtain the following representation of $\tau$: 
\begin{equation} \label{SignedEq}
\tau = \left( 1 + q \right) {\sigma_2}_{\vert_{\Sigma_{\boldsymbol{\gamma}}}} + \left( 1 + q \right) \sum_{i=1}^m  \widehat{\nu_i} - \sum_{i=1}^m q_i \, \widehat{\delta_{\PT{a}_i}}.
\end{equation}
%
%
%
%

We show that the weighted logarithmic potential of $\tau$ satisfies \eqref{signedeq}. 
Let $\PT{x} \in \Sigma_{\boldsymbol{\gamma}}$. Then $\PT{x} \in \Sigma_{i, \gamma_i}$ for every $1 \leq i \leq m$ and, by \eqref{LogBal} and \eqref{BalDelta}, for every $1 \leq i \leq m$ 
\begin{align*}
U_0^{\widehat{\nu_i}}( \PT{x} ) &= U_0^{\nu_i}( \PT{x} ) + C_i \qquad \text{on $\Sigma_{i,\gamma_i}$,} \\
U_0^{\widehat{\delta_{\PT{a}_i}}}( \PT{x} ) &= U_0^{\delta_{\PT{a}_i}}( \PT{x} ) - \frac{1}{2} \, \log \frac{1+t_i}{2} \qquad \text{on $\Sigma_{i,\gamma_i}$.} 
\end{align*}
%
%
%
Hence, computing the logarithmic potential of $\tau$ in \eqref{SignedEq} yields, after simplification,
\begin{equation*}
U_0^\tau( \PT{x} ) + Q( \PT{x} ) = \left( 1 + q \right) U_0^{\sigma_2}( \PT{x} ) + \sum_{i=1}^m \left( \left( 1 + q \right) C_i + \frac{q_i}{2} \, \log \frac{1+t_i}{2} \right), \qquad \PT{x} \in \Sigma_{\boldsymbol{\gamma}}.
\end{equation*}
Since $U_0^{\sigma_2}( \PT{x} ) = W_0( \mathbb{S}^2 )$, the weighted potential of $\tau$ is constant on $\Sigma_{\boldsymbol{\gamma}}$; i.e., $\tau$ is a signed equilibrium on $\Sigma_{\boldsymbol{\gamma}}$ associated with $Q$ and, by uniqueness, $\eta_{\Sigma_{\boldsymbol{\gamma}},Q} = \tau$ and
\begin{equation*}
F_{\Sigma_{\boldsymbol{\gamma}},Q} = \left( 1 + q \right) W_0( \mathbb{S}^2 ) + \sum_{i=1}^m \left( \left( 1 + q \right) C_i + \frac{q_i}{2} \, \log \frac{1+t_i}{2} \right).
\end{equation*}

Let $\PT{x} \in \Sigma_{\boldsymbol{\gamma}}^c$. Then $\PT{x} \in \Sigma_{i_0,\gamma_{i_0}}^c$ for some $i_0 \in \{ 1, \dots, m \}$ and $\PT{x} \in \Sigma_{i,\gamma_{i}}$ for $i \neq i_0$. Using \eqref{SignedEq}, \eqref{BalSigma}, and \eqref{BalDelta}, 
\begin{equation} \label{eq:weighted.potential.form.A}
\begin{split}
U_0^{\eta_{\Sigma_{\boldsymbol{\gamma}},Q}}( \PT{x} ) + Q( \PT{x} ) = F_{\Sigma_{\boldsymbol{\gamma}},Q} &+ \left( 1 + q \right) \left[ U_0^{\widehat{{\nu}_{i_0}}}( \PT{x} ) - \left( U_0^{{\nu}_{i_0}}( \PT{x} ) + C_{i_0} \right) \right] \\
&+ \frac{q_{i_0}}{2} \log \frac{\left( 1 + \xi_{i_0} \right) \left( 1 - t_{i_0} \right)}{\left( 1 - \xi_{i_0} \right) \left( 1 + t_{i_0} \right)}.
\end{split}
\end{equation}
Observe that the square-bracketed expression is $\leq 0$ by \eqref{LogBal}. Because of ${t_{i_0} < \xi_{i_0} < 1}$, the ratio under the logarithm is $>1$ and the logarithm tends to zero as $\xi_{i_0}$ goes to~$t_{i_0}$ and the logarithm tends to $+\infty$ as $\xi_{i_0}$ approaches $1$ from below. Using \eqref{BalSigma} again, we derive
\begin{equation}
\begin{split}
U_0^{\eta_{\Sigma_{\boldsymbol{\gamma}},Q}}( \PT{x} ) + Q( \PT{x} ) 
&= \left( 1 + q \right) W_0( \Sigma_{i_0,\gamma_{i_0}} ) + \sum_{\substack{i=1,\\ i\neq i_0}}^m \left( \left( 1 + q \right) C_i + \frac{q_i}{2} \, \log \frac{1+t_i}{2} \right) \\
&\phantom{=}+ \frac{q_{i_0}}{2} \, \log \frac{1+t_{i_0}}{2} + f( t_{i_0} ) - f( \xi_{i_0} ),
\end{split}
\end{equation}
where
\begin{equation} \label{f-function} 
f(u) := \frac{1+q-q_{i_0}}{2} \log(1+u) + \frac{q_{i_0}}{2} \log(1-u), \qquad -1 < u < 1.
\end{equation}
The function $f$ has a unique maximum at $u^* = 1 - \frac{2 q_{i_0}}{1+q}$ in the interval $(-1,1)$ for $0 < q_{i_0} < 1 + q$. 
Assuming that $-1 < t_{i_0} < u < 1$, $f^\prime(u) < 0$ if and only if 
\begin{equation*}
\max\Big\{ t_{i_0}, 1 -  \frac{2 q_{i_0}}{1+q} \Big\} < u < 1 \  \Leftrightarrow \  0 < 2( 1 - u ) < \min\Big\{ \frac{4q_{i_0}}{1+q} , \gamma_{i_0}^2 \Big\} = \min\Big\{ \epsilon_{i_0}^2 , \gamma_{i_0}^2 \Big\}.
\end{equation*}
By assumption, $\gamma_{i_0} \leq \epsilon_{i_0}$. Hence, the infimum of the weighted potential of $\eta_{\Sigma_{\boldsymbol{\gamma}},Q}$ in the set $\Sigma_{i_0,\gamma_{i_0}}^c$ is assumed on its boundary. 
 Continuity of the potentials in \eqref{eq:weighted.potential.form.A} yields
\begin{equation*}
U_0^{\eta_{\Sigma_{\boldsymbol{\gamma}},Q}}( \PT{x} ) + Q( \PT{x} ) \geq F_{\Sigma_{\boldsymbol{\gamma}},Q} \qquad \text{on $\Sigma_{i_0,\gamma_{i_0}}^c$.}
\end{equation*}
As $i_0$ was determined by $\PT{x} \in \Sigma_{\boldsymbol{\gamma}}^c$, we deduce that the last relation holds on $\Sigma_{\boldsymbol{\gamma}}^c$.

Summarizing, for each $\boldsymbol{\gamma} \leq \boldsymbol{\epsilon}$
\begin{align}
U_0^{\eta_{\Sigma_{\boldsymbol{\gamma}},Q}}( \PT{x} ) + Q( \PT{x} ) 
&\geq F_{\Sigma_{\boldsymbol{\gamma}},Q} \qquad \text{on $\Sigma_{\boldsymbol{\gamma}}^c$,} \label{SignedVarEq1}\\
U_0^{\eta_{\Sigma_{\boldsymbol{\gamma}},Q}}( \PT{x} ) + Q( \PT{x} ) 
&= F_{\Sigma_{\boldsymbol{\gamma}},Q} \qquad \text{on $\Sigma_{\boldsymbol{\gamma}}$} \label{SignedVarEq2} 
\end{align}
%
and from \eqref{SignedEq} and \eqref{BalLog},
\begin{equation*} \label{SignedEqB}
\eta_{\Sigma_{\boldsymbol{\gamma},Q}} = \left( 1 + q \right) {\sigma_2}_{\vert_{\Sigma_{\boldsymbol{\gamma}}}} + \left( 1 + q \right) \sum_{i=1}^m \left( \frac{\gamma_{i}^2}{4} - \frac{\epsilon_i^2}{4} \right) \beta_i.
\end{equation*}
It is not difficult to see that the signed equilibrium $\eta_{\Sigma_{\boldsymbol{\gamma},Q}}$ becomes a positive measure, and at the same time satisfies the characterization inequalities \eqref{VarEq3} and \eqref{VarEq4}, if and only if $\boldsymbol{\gamma} = \boldsymbol{\epsilon}$. By Proposition~\ref{prop:1}(d), $\mu_Q = \eta_{\Sigma_{\boldsymbol{\epsilon},Q}} = \left( 1 + q \right) {\sigma_2}_{\vert \Sigma_{\boldsymbol{\epsilon}}}$.
%
%
%
%
\end{proof}

Theorem \ref{thm:main} and \cite[Corollary 13]{BraDraSaf2014} yield the following result.


\begin{corollary} \label{Cor} 
Under the assumptions of Theorem~\ref{thm:main}, the optimal logarithmic energy $N$-point configurations w.r.t. $Q$ are contained in $S_Q$ for every $N \geq 2$.  
\end{corollary}

\begin{proof}
From \cite[Corollary 13]{BraDraSaf2014} we have that the optimal $N$-point configurations lie in 
\begin{equation*}
\widetilde{S}_Q=\{\PT{x} \ :\ U_0^{\mu_Q} (\PT{x})+Q(\PT{x}) \leq F_Q\}.
\end{equation*}
The strict monotonicity of the function $f$ in \eqref{f-function} yields $\widetilde{S}_Q = S_Q$.
\end{proof}

\begin{figure}[ht]
\begin{center}
\includegraphics[scale=.21]{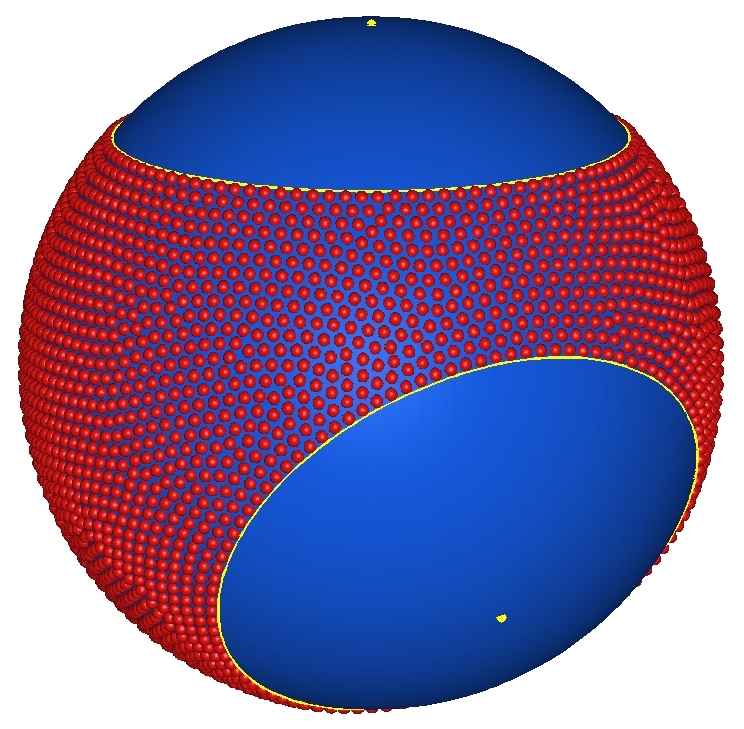} \qquad \includegraphics[scale=.21]{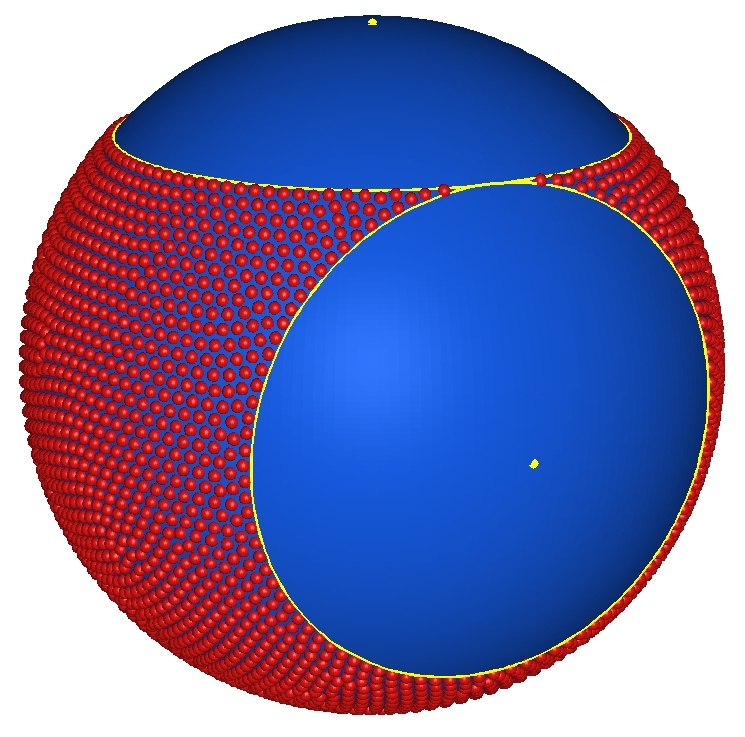} 
\end{center}
\caption{\label{fig1} Approximate log-optimal points for $m=2$, $N=4000$ with
$q_1 = q_2 = \frac{1}{4}$, $\PT{a}_1 = ( 0, 0, 1 )$ and
$\PT{a}_2 = ( \frac{\sqrt{91}}{10}, 0, -\frac{3}{10} )$ or
$\PT{a}_2 = ( \frac{4\sqrt{5}}{9}, 0, -\frac{1}{9} )$}
\end{figure}

\begin{remark} Theorem \ref{thm:main} and Corollary \ref{Cor} are illustrated in Figures \ref{fig1} and \ref{fig2} 
for two and three point sources, respectively. Observe, that the density of the (approximate) log-optimal configuration approaches the normalized surface area of the equilibrium support $S_Q=\Sigma_{\boldsymbol{\epsilon}}$. 
\end{remark}

\begin{figure}[ht]
\begin{center}
\hspace{-12mm}\includegraphics[scale=.22]{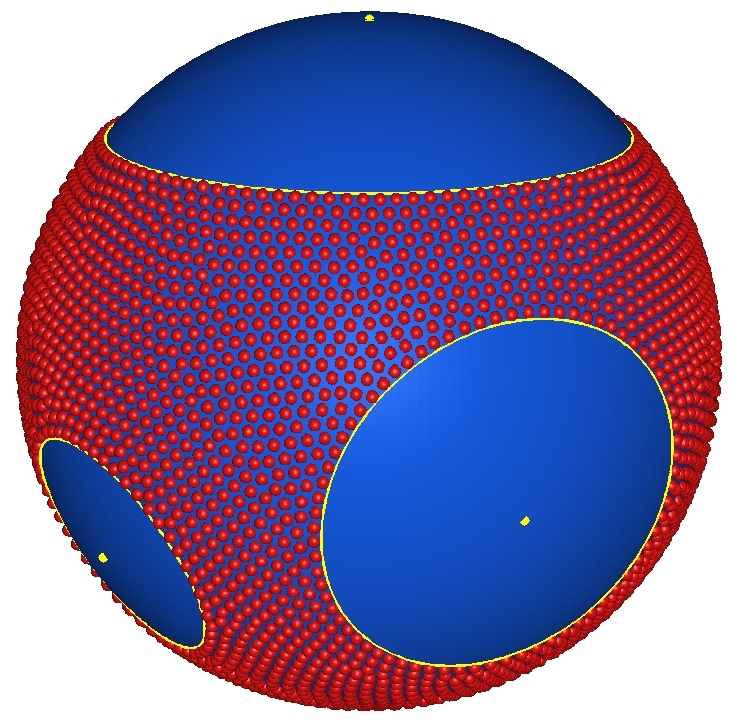} 
\end{center}
\caption{\label{fig2} Approximate log-optimal points for $m=3$, $N=4000$ with
$q_1 = \frac{1}{4}$, $q_2 = \frac{1}{8}$, $q_3 = \frac{1}{20}$, $\PT{a}_1=(0,0,1)$,
$\PT{a}_2 = ( \frac{\sqrt{91}}{10}, 0, -\frac{3}{10} )$, and
$\PT{a}_3 = ( 0, \frac{\sqrt{3}}{2}, -\frac{1}{2} )$}
\end{figure}

\begin{remark}
The objective function for the optimization problem (\ref{Eq})
with the discrete external field (\ref{DiscrField}) is
\[
  E_{Q,N}(\PT{x}_1,\ldots,\PT{x}_N) =
  \sum_{1\leq i \neq j\leq N} k_s(\PT{x}_i,\PT{x}_j) +
  2 (N-1) \sum_{i=1}^m q_i \sum_{j=1}^N k_{s_i}(\PT{a}_i,\PT{x}_j),
\]
where $k(\PT{x},\PT{y})$ is the Riesz kernel defined at he beginning
of Section~1.
The standard spherical parametrisation,
$\PT{x}_i = (\sin(\theta_i)\cos(\phi_i) \; \sin(\theta_i)\sin(\phi_i) \; \cos(\theta_i))^T \in \mathbb{S}^{2}$
for $\theta_i\in[0, \pi]$ and $\phi_i\in [0, 2\pi)$
is used to avoid the non-linear constraints $\PT{x}_i \cdot \PT{x}_i = 1, i = 1,\ldots,N$.
This introduces singularities at the poles $\theta = 0, \pi$,
one of which can be avoided by using the rotational invariance of the
objective function to place the first external field at the north pole.
For $\theta_i \neq 0, \pi$ the gradient of $E_{Q,N}(\theta_1, \phi_1, \ldots, \theta_N, \phi_N)$
can be calculated for use in a nonlinear optimization method.

Point sets $\{\PT{x}_1,\ldots,\PT{x}_N\}$ that provide
approximate optimal s-energy configurations
were obtained using this spherical parametrisation of the points
and applying a nonlinear optimization method,
for example a limited memory BFGS method for
bound constrained problems~\cite{ZhuByrLuNoc1997},
to find a local minimum of $E_{Q,N}$.
The initial point sets $\{\PT{x}_1,\ldots,\PT{x}_N\}$
used as starting points for the nonlinear optimization
were uniformly distributed on $\mathbb{S}^2$,
so did not reflect the structure of the external fields.
A local perturbation of the point set achieving a local minimum was then used
to generate a new starting point and the nonlinear optimization applied again.
The best local minimizer found provided an approximation (upper bound)
on the global minimum of $E_{Q,N}$.
Different local minima arose from the fine structure of the points within their support.
\end{remark}

%
%
%
%
%
%
%
%
%

The results above lend themselves to the following generalization. Given $m$ points $\PT{a}_1,\dots,\PT{a}_m \in \mathbb{S}^2$, for each $i=1,\dots,m$ let $\phi_i$ be a  radially-symmetric  measure centered at $\PT{a}_i$ and supported on $\Sigma_{i,\rho_i}^c$ for some $\rho_i>0$ that has absolutely continuous density with respect to $\sigma_2$; i.e.,
\begin{equation} \label{Phi}
\D \phi_i( \PT{x} ) = f_i(\langle \PT{x}, \PT{a}_i \rangle ) \, \D \sigma_2( \PT{x} ), \qquad f_i (u) = 0 \quad \text{on $\bigg[ -1, \sqrt{1-\rho_i^2/2} \bigg]$.}
\end{equation}
Let $q_i := \| \phi_i \| = \int \D \phi_i( \PT{x} )$, $1 \leq i \leq m$, and define the external field 
\begin{equation}\label{PhiField}
Q_{\boldsymbol{\phi}}( \PT{x} ) := \sum_{i=1}^m U_0^{ \phi_i}(\PT{x}) = \sum_{i=1}^m \int\log \frac{1}{|\PT{x}-\PT{a}_i |} \, \D \phi_i (\PT{x}),
\end{equation}
where $\boldsymbol{\phi} = ( \phi_1, \dots, \phi_m )$.
Then the following theorem holds.
%
\begin{theorem} \label{thm:main---}
Let $d = 2$ and $s = 0$. Let $Q_{\boldsymbol{\phi}}$ be defined by \eqref{PhiField} and $\boldsymbol{\epsilon}$, $\Sigma_{\boldsymbol{\epsilon}}$ be defined by \eqref{epsilon}, \eqref{Sigma_epsilon}. Suppose that $\Sigma_{i,\epsilon_i}^c \cap \Sigma_{j,\epsilon_j}^c = \emptyset$, $1 \leq i < j \leq m$. Then the logarithmic extremal measure associated with $Q_{\boldsymbol{\phi}}$ is $\mu_{Q_{\boldsymbol{\phi}}} = \left( 1 + q \right) {\sigma_2}_{\vert_{\Sigma_{\boldsymbol{\epsilon}}}}$ and the extremal support is $S_{Q_{\boldsymbol{\phi}}} = \Sigma_{\boldsymbol{\epsilon}}$.
\end{theorem}

\begin{proof} The proof proceeds as in the proof of Theorem~\ref{thm:main} with the adaption that the balayage measure of $\phi_i$ is given by
\begin{equation*}
\widehat{\phi}_i = \bal_0(\phi_i,\Sigma_{i,\epsilon_i}) = \| \phi \| \, \beta_i = q_i \, \beta_i,
\end{equation*}
which follows easily from the hypothesis $\rho_i \leq \epsilon_i$ and the uniqueness of balayage measures. 
\end{proof}

\vskip 2mm

We next formulate the analog  of Theorem \ref{thm:main} in the complex plane $\mathbb{C}$. Let us fix one of the charges, say $\PT{a}_m$, at the North Pole $\PT{p}$, which will also serve as the center of the \emph{Kelvin transformation} $\mathcal{K}$ (stereographic projection, or equivalently, inversion about the center $\PT{p}$) with radius $\sqrt{2}$ onto the equatorial plane. 
Set $w_i := \mathcal{K}(\PT{a}_i)$, $1 \leq i \leq m$. The image of $\PT{a}_m$ under the Kelvin transformation is the "point at infinity" in $\mathbb{C}$.
Letting $z = \mathcal{K}(\PT{x})$, $\PT{x} \in \mathbb{S}^2$, we can utilize the following formulas
\begin{equation*}
\left| \PT{x} - \PT{p} \right| = \frac{2}{\sqrt{1+|z|^2}}, \qquad \left| \PT{x} - \PT{a}_i \right| = \frac{2 |z - w_i |}{\sqrt{ 1 + |z|^2 } \, \sqrt{ 1 + |w_i|^2} }, \quad 1 \leq i \leq m-1
\end{equation*}
to convert the continuous minimal energy problem (cf. \eqref{Vq}) and the discrete minimal energy problem (cf. \eqref{Eq}) on the sphere to their analog forms  in the complex plane $\mathbb{C}$.
Neglecting a constant term, we obtain in the complex plane the external field
\begin{equation}\label{DiscrFieldC}
\widetilde{Q}(z):=\sum_{i=1}^{m-1}q_i\log \frac{1}{|z-w_i|}+\left( 1 + q \right)\log \sqrt{1+|z|^2}, \qquad z \in \mathbb{C}. 
\end{equation}
This external field is admissible in the sense of Saff-Totik \cite{SafTot1997}, since
\begin{equation*}
\lim_{|z| \to \infty} \left( \widetilde{Q}(z) - \log |z| \right) = \lim_{|z| \to \infty} q_m \log |z| = \infty.
\end{equation*}
Therefore, there is a unique equilibrium measure $\mu_{\widetilde{Q}}$ characterized by variational inequalities similar to the ones in Proposition~\ref{prop:1}(d). The following theorem giving the extremal support $S_{\widetilde{Q}}$ and the extremal measure $\mu_{\widetilde{Q}}$ associated with the external field $\widetilde{Q}$ in \eqref{DiscrFieldC} for sufficiently small $q_i$'s is a direct consequence of Theorem \ref{thm:main}. 

\begin{theorem} \label{thm:mainC}
Let $w_1, \dots, w_{m-1} \in \mathbb{C}$ be fixed and $q_1, \dots, q_m$ be positive real numbers with $q = q_1 + \cdots + q_m$ and $\widetilde{Q}$ be the corresponding external field given in \eqref{DiscrFieldC}. 
Further, let $\PT{a}_1, \dots, \PT{a}_{m-1} \in \mathbb{S}^2$ be the pre-images under the Kelvin transformation~$\mathcal{K}$, i.e., $w_i = \mathcal{K}( \PT{a}_i )$, $1 \leq i \leq m-1$, and $a_{m} = \PT{p}$. If the $q_i$'s are sufficiently small so that $\Sigma_{i,\epsilon_i}^c \cap \Sigma_{j,\epsilon_j}^c = \emptyset$, $1\leq i <j \leq m$, where the spherical caps $\Sigma_{i,\epsilon_i}$ are defined in \eqref{Sigma_epsilon}, then there are open discs $D_1, \dots, D_{m-1}$ in $\mathbb{C}$ with $w_i \in D_i=\mathcal{K}(\Sigma_{i,\epsilon_i})$, $1 \leq i \leq m-1$, such that 
\begin{equation} \label{EqSuppC}
S_{\widetilde{Q}} = \left\{ z \in \mathbb{C} \, : \, |z| \leq \sqrt{\frac{1+q-q_m}{q_m}} \right\} \setminus \bigcup_{i=1}^{m-1} D_i .
\end{equation}
The extremal measure $\mu_{\widetilde{Q}}$ associated with $  \widetilde{Q}$ is given by
\begin{equation} \label{EqMeC}
\D \mu_{\widetilde{Q}}(z) = \frac{1+q}{\pi \left( 1 + |z|^2 \right)^2} \, \D A(z),
\end{equation}
where $\D A$ denotes the Lebesgue area measure in the complex plane. 
\end{theorem}

\begin{proof} The proof follows by a straight forward application of the Kelvin transformation to the weighted potential $U_0^{\mu_Q} (\PT{x}) +Q(\PT{x})$ and using the identity relating the regular (not normalized) Lebesgue measure on the sphere and the area measure on the complex plane
\begin{equation*}
\frac{4\pi}{\left| \PT{x} - \PT{p} \right|^2} \, \D \sigma_2( \PT{x} ) = \frac{1}{\left| z - \PT{p} \right|^2} \, \D A(z)
\end{equation*}
This change of variables yields the identity
\begin{equation*}
U_0^{\mu_Q} (\PT{x}) +Q(\PT{x})=U_0^{\mu_{\widetilde{Q}}} (z) + \widetilde{Q}(z) + const
\end{equation*}
from which, utilizing \eqref{SignedVarEq1} and \eqref{SignedVarEq2}, one derives 
\begin{align}
U_0^{\mu_{\widetilde{Q}}}(z) + \widetilde{Q}(z) &\geq C \quad \text{in $\mathbb{C}$,} \label{VarEqC1} \\
U_0^{\mu_{\widetilde{Q}}}(z) + \widetilde{Q}(z) &= C \quad \text{on $S_{\widetilde{Q}}$,} \label{VarEqC2}
\end{align}
which implies that $\mu_{\widetilde{Q}}$ is the equilibrium measure by \cite[Theorem 1.3]{SafTot1997}.
 \end{proof}


\begin{remark}
At first it seems like a surprising fact that the equilibrium measure in Theorem \ref{thm:main---} is uniform on $S_Q$ (i.e. has constant density). However, this can be easily seen alternatively from the planar version Theorem \ref{thm:mainC}. Once we derive that the support $ S_{\widetilde{Q}} $
is given by \eqref{EqSuppC}, we can recover the measure $ \mu_{\widetilde{Q}}$ by applying
Gauss'  theorem (cf. \cite[Theorem II.1.3]{SafTot1997}), namely on any subregion of $ S_{\widetilde{Q}} $ we have 
\[d\mu_{\widetilde{Q}} = - \frac{1}{2\pi} \Delta U^{\mu_{\widetilde{Q}}} dA(z) =  \frac{1}{2\pi} \Delta \widetilde{Q}(z) =\frac{ 1+q}{\pi(1 +|z|^2)^2} \, dA(z).\]
Recall that on this subregion $\log |z-w_i|$ is harmonic for all $i=1,\dots,m-1$. As $d\sigma_2 (x)= dA(z)/[\pi(1 +|z|^2)^2] $, 
we get that $\mu_Q$ is the normalized Lebesgue surface measure on $S_Q$. Observe, that the same argument applies to the setting of Theorem \ref{RegionInfluence} ($d=2,\ s=0$), from which we derive $\mu_Q = \left( 1 + q \right) {\sigma_2}_{\vert_{S_Q}}$ even in the case when $\Sigma_{\epsilon_i}^c$ are not disjoint. Of course, we don't know the equilibrium support $S_Q$ in this case. For related results see \cite{Bel2013, Bel2015}.
%
\end{remark}

\section{Riesz $(d-2)$-Energy Interactions on $\mathbb{S}^d$, $d \geq 3$} 
\label{Rie}

The case of $(d-2)$-energy interactions on $\mathbb{S}^d$, $d\geq 3$, and an external field $Q$ given by~\eqref{DiscrField} is considerably more involved as the balayage measures utilized to determine the signed equilibrium on $\Sigma_{\boldsymbol{\gamma}}$ diminish their masses. This phenomenon yields an implicit nonlinear system for the critical values of the radii $\epsilon_1, \dots, \epsilon_m$ (see \eqref{eq:gamma.i.cond.} and \eqref{EquationC}) characterizing the regions of electrostatic influence. 
%

Let $d \geq 3$ and $0 < d - 2 \leq s < d$. Let $\Phi_s(t_i) := \mathcal{F}_s( \Sigma_{i, \gamma_i} )$ be the Mhaskar-Saff $\mathcal{F}$-functional associated with the external field $Q_i (\PT{x}) := q_i \, |\PT{x} - \PT{a}_i |^{-s}$ evaluated for the spherical cap $\Sigma_{i, \gamma_i}$. Then the signed $s$-equilibrium measure $\eta_{i,s} := \eta_{\Sigma_{i, \gamma_i},Q_i,s}$ on $\Sigma_{i, \gamma_i}$ associated with $Q_i$ is given by (see \cite[Theorem 11 and 15]{BraDraSaf2009})
%
%
%
\begin{equation} \label{eta_{t_i}}
\eta_{i,s} = \frac{\Phi_s(t_i)}{W_s( \mathbb{S}^d ) } \, \bal_s(\sigma_d, \Sigma_{i,\gamma_i}) - q_i \, \bal_s ( \delta_{\PT{a}_i}, \Sigma_{i,\gamma_i}). 
\end{equation}
For $d-2<s<d$ this signed measure is absolutely continuous 
%
\begin{equation*}
\D \eta_{i,s}( \PT{x} ) = \frac{\omega_{d-1}}{\omega_{d}} \eta_{i,s}^{\prime}(u) \left( 1 - u^2 \right)^{d/2-1} \, \dd u \, \dd\sigma_{d-1}(\overline{\PT{x}}), \quad \PT{x} = ( \sqrt{1-u^2} \, \overline{\PT{x}}, u) \in \Sigma_{i,\gamma_i},
\end{equation*}
with density function 
\begin{equation}
\begin{split} \label{eta.t.prime.1st.result}
\eta_{i,s}^{\prime}(u) 
&= \frac{1}{W_s(\mathbb{S}^d)} \frac{\gammafcn(d/2)}{\gammafcn(d-s/2)}
\left( \frac{1-t_i}{1-u} \right)^{d/2} \left( \frac{t_i-u}{1-t_i} \right)^{(s-d)/2} \\
&\phantom{=\times}\times \Bigg\{ \Phi_s(t_i) \, \HypergeomReg{2}{1}{1,d/2}{1-(d-s)/2}{\frac{t_i-u}{1-u}} - \frac{q_i \, 2^{d-s}}{{\gamma_i}^{d}} \Bigg\}. 
\end{split}
\end{equation}
For the ratio $\frac{\omega_{d-1}}{\omega_d}$ see~\eqref{omega.ratio}, a formula for the Riesz $s$-energy $W_0( \mathbb{S}^d )$ is given in \eqref{SphereEnergy}, and $_2\HyperTildeF_1$ denotes Olver's regularized $_2F_1$-hypergeometric function \cite[Eq.~15.2.2]{NIST:DLMF}.
For $s = d - 2$ the signed $(d-2)$-equilibrium
\begin{equation} \label{etabar}
\eta_{i,d-2}= \frac{\Phi_{d-2}(t_i)}{W_{d-2}(\mathbb{S}^d)}  {\sigma_{d}}_{\vert_{\Sigma_{i,\gamma_i}}}
+ \frac{1-t_i}{2} \left( 1 - t_i^2 \right)^{d/2-1} \left[ \Phi_{d-2}(t_i) - \frac{4q_i}{\gamma_i^d} \right] \beta_{i}
\end{equation}
has, like in the logarithmic case (see \eqref{eq:eta.simple}), a boundary-supported component $\beta_i$, which is the normalized Lebesgue measure on the boundary circle of $\Sigma_{i,\gamma}$.
%
Observe that in either case the signed equilibrium has a negative component if and only if 
\begin{equation} \label{eq:negativity.relation}
\Phi_s(t_i) - \frac{2^{d-s} q_i}{\gamma_i^d} < 0, \qquad \text{where $2( 1 - t_i ) = \gamma_i^2$.}
\end{equation}
%
The weighted $s$-potential of $\eta_{i,s}$, $d-2 < s < d$, satisfies (\cite[Theorem~11]{BraDraSaf2009})
\begin{align}
U_s^{\eta_{i,s}}(\PT{z})+Q_i (\PT{z}) &= \Phi_s(t_i), \qquad \PT{z} \in \Sigma_{i,\gamma_i},  \label{SignedPot1}\\
\begin{split}
U_s^{\eta_{i,s}}(\PT{z})+Q_i (\PT{z}) &= \Phi_s(t_i) + \frac{q_i }{[2(1-\xi_i )]^{s/2}} \, \IncompleteBetaRegularized\Bigg(\frac{2}{1-t_i} \frac{\xi_i-t_i}{1+\xi_i}; \frac{d-s}{2}, \frac{s}{2} \Bigg) \\
&\phantom{=\pm}- \Phi_s(t_i) \, \IncompleteBetaRegularized\Bigg(\frac{\xi_i-t_i}{1+\xi_i};
\frac{d-s}{2}, \frac{s}{2}\Bigg), \qquad \PT{z} \in \mathbb{S}^d \setminus \Sigma_{i,\gamma_i}, \label{eq:weighted.outside}
\end{split}
\end{align}
where $\PT{z} = ( \sqrt{1-\xi_i^2}\; \overline{\PT{z}},\xi_i)\in \mathbb{S}^d$, $-1 \leq \xi_i \leq 1$ and $\overline{\PT{z}} \in \mathbb{S}^{d-1}$, and 
\begin{equation*}
\IncompleteBetaRegularized(x;a,b) := \frac{\betafcn(x;a,b)}{\betafcn(a,b)}, \quad \betafcn(a,b) := \betafcn(1;a,b), \quad \betafcn(x;a,b) := \int_0^x u^{a-1} (1-u)^{b-1}\, du
\end{equation*}
are the regularized incomplete beta function, the beta function, and the incomplete beta function \cite[Ch.~5 and 8]{NIST:DLMF}; whereas (\cite[Lemmas~33 and 36]{BraDraSaf2009})
\begin{align}
U_{d-2}^{\eta_{i,d-2}}(\PT{z}) + Q_i(\PT{z}) &= \Phi_{d-2}(t_i), \qquad \PT{z} \in \Sigma_{i,\gamma_i},  \label{SignedPot1special}\\
\begin{split}
U_{d-2}^{\eta_{i,d-2}}(\PT{z}) + Q_i(\PT{z}) &= \Phi_{d-2}(t_i) \left( \frac{1+t_i}{1+\xi_i} \right)^{d/2-1} + \frac{q_i}{\left( 2 ( 1 - \xi_i ) \right)^{d/2-1}} \\
&\phantom{=\pm}- \frac{q_i}{\gamma_i^{d-2}} \left( \frac{1+t_i}{1+\xi_i} \right)^{d/2-1}, \qquad \PT{z} \in \mathbb{S}^d \setminus \Sigma_{i,\gamma_i}. \label{eq:weighted.outside.special}
\end{split}
\end{align}
The last relation follow from \eqref{eq:weighted.outside} if $s$ is changed to $d-2$.

In the proof of our main result for $s=d-2$, $d\geq 3$, we need the analog of \eqref{SignedVarEq1}, which we derive from a similar result for the weighted potential \eqref{eq:weighted.outside}. As this is of independent interest, we state and prove the following lemma for $d-2 \leq s < d$.

\begin{lemma} \label{Lem} 
Let $d \geq 3$ and $d - 2 \leq s < d$. If \eqref{eq:negativity.relation} is satisfied, then the weighted $s$-potential of the signed $s$-equilibrium  $\eta_{i,s}$ satisfies the variational inequalities
\begin{align}
U_s^{\eta_{i,s}}(\PT{z}) + Q_i(\PT{z}) &= \Phi_s(t_i), \qquad \PT{z} \in \Sigma_{i,\gamma_i}, \label{Signed1} \\
U_s^{\eta_{i,s}}(\PT{z}) + Q_i(\PT{z}) &> \Phi_s(t_i),\qquad \PT{z} \in \mathbb{S}^d \setminus \Sigma_{i,\gamma_i}. \label{Signed2}
\end{align}
Furthermore, both relations remain valid if equality is allowed in \eqref{eq:negativity.relation}.
\end{lemma}

\begin{proof} The first equality \eqref{Signed1} was established in \cite[Theorems 11 and 15]{BraDraSaf2009}. 

Let $d \geq 3$ and $d - 2 \leq s < d$. The right-hand side of \eqref{eq:weighted.outside} is a function of $\xi_i$ with $t_i < \xi_i \leq 1$. We denote it by $G( \xi_i )$. Using the integral form of the incomplete regularized beta function, we get 
\begin{equation*}
\begin{split}
\betafcn\Big( \frac{d-s}{2}, \frac{s}{2} \Big) \, \Big( G( \xi_i ) - \Phi_s( t_i ) \Big) 
&= \left( \frac{1-t_i}{1-\xi_i} \right)^{s/2} \frac{2^{d-s} q_i}{\gamma_i^{d}} \int_0^{\frac{\xi_i-t_i}{1+\xi_i}} u^{\frac{d-s}{2}-1} \left( 1 - \frac{2 \, u}{1-t_i} \right)^{\frac{s}{2} - 1} \dd u \\
&\phantom{=}- \Phi_s( t_i ) \int_0^{\frac{\xi_i-t_i}{1+\xi_i}} u^{\frac{d-s}{2}-1} \left( 1 - u \right)^{\frac{s}{2} - 1} \dd u.
\end{split}
\end{equation*}
Let \eqref{eq:negativity.relation} be satisfied. Then
\begin{equation*}
\begin{split}
\betafcn\Big( \frac{d-s}{2}, \frac{s}{2} \Big) \, \frac{G( \xi_i ) - \Phi_s( t_i )}{\Phi_s( t_i )} 
&> \left[ \frac{1-t_i}{1-\xi_i} \right]^{s/2} \int_0^{\frac{\xi_i-t_i}{1+\xi_i}} u^{\frac{d-s}{2}-1} \left( 1 - \frac{2}{1-t_i} \, u \right)^{\frac{s}{2} - 1} \dd u \\
&\phantom{=}- \int_0^{\frac{\xi_i-t_i}{1+\xi_i}} u^{\frac{d-s}{2}-1} \left( 1 - u \right)^{\frac{s}{2} - 1} \dd u.
\end{split}
\end{equation*}
The square-bracketed expression is $>1$ for $-1 < t_i < \xi_i \leq 1$. Since ${\frac{2}{1-t_i} > 1}$, the first integrand is bounded from below by the second integrand if $\frac{s}{2} - 1 \leq 0$. In the case $\frac{s}{2} - 1 > 0$, we observe that for $0 \leq u \leq \frac{\xi_i - t_i}{1+\xi_i}$, 
\begin{equation*}
\begin{split}
\left[ \frac{1-t_i}{1-\xi_i} \right]^{s/2} \left( 1 - \frac{2}{1-t_i} \, u \right)^{\frac{s}{2} - 1} 
&= \frac{1-t_i}{1-\xi_i} \left( \frac{1-t_i}{1-\xi_i} - \frac{2}{1-\xi_i} \, u \right)^{\frac{s}{2} - 1} \\
&> \left( \frac{1-t_i}{1-\xi_i} - \frac{2}{1-\xi_i} \, u \right)^{\frac{s}{2} - 1} \\
&\geq \left( 1 - u \right)^{\frac{s}{2} - 1}.
\end{split}
\end{equation*}
The estimates are strict in both cases. Hence, equality is allowed in \eqref{eq:negativity.relation}.
\end{proof}

We are now ready to state and prove the second main result.

\begin{theorem} \label{thm:main2}
Let $d\geq 3$ and $s = d - 2$. Let $Q$ be defined by \eqref{DiscrField}.
Suppose the positive charges $q_1, \dots, q_m$ are sufficiently small. Then there exists a critical $\boldsymbol{\epsilon} = ( \epsilon_1, \dots, \epsilon_m )$, uniquely defined by these charges, such that $\Sigma_{\epsilon_i}^c \cap \Sigma_{\epsilon_j}^c = \emptyset$, $1\leq i <j \leq m$, and
the $(d-2)$-extremal measure associated with $Q$ is $\mu_Q = C \, {\sigma_d}_{\vert_{\Sigma_{\boldsymbol{\epsilon}}}}$ for a uniquely defined normalization constant $C > 1$ and the extremal support is $S_Q = \Sigma_{\boldsymbol{\epsilon}}$.

Furthermore, an optimal $(d-2)$-energy $N$-point configuration w.r.t. $Q$ is contained in $S_Q$ for every $N \geq 2$. 
\end{theorem}

\begin{proof} Let $\boldsymbol{\gamma} = ( \gamma_1, \dots, \gamma_m )$ be a vector of $m$ positive numbers such that ${\Sigma_{i,\gamma_i}^c \cap \Sigma_{i,\gamma_j}^c = \emptyset}$, $1\leq i <j \leq m$. 
We consider the signed measure 
\begin{equation*}
\tau := C \, \bal_{d-2}( \sigma_d, \Sigma_{\boldsymbol{\gamma}} ) - \bal_{d-2}( q_1 \, \delta_{\PT{a}_1} + \cdots + q_m \, \delta_{\PT{a}_m}, \Sigma_{\boldsymbol{\gamma}} ).  
\end{equation*}
As balayage under Riesz $(d-2)$-kernel interactions satisfies \eqref{sBal}, we have
\begin{align*}
U_{d-2}^{\tau}( \PT{z} ) + Q( \PT{z} ) &= C \, U_{d-2}^{\sigma_d}( \PT{z} ) = C \, W_{d-2}( \mathbb{S}^d ), & \PT{z} &\in \Sigma_{\boldsymbol{\gamma}}, \\
U_{d-2}^{\tau}( \PT{z} ) + Q( \PT{z} ) &\geq C \, U_{d-2}^{\bal_{d-2}( \sigma_d, \Sigma_{\boldsymbol{\gamma}} )}( \PT{z} ), & \PT{z} &\in \mathbb{S}^d \setminus \Sigma_{\boldsymbol{\gamma}}.
\end{align*}
If the normalization constant $C = C( \boldsymbol{ \gamma } )$ is chosen such that $\| \tau \| = \tau( \Sigma_{\boldsymbol{\gamma}} ) = 1$, then $\tau$ is a signed $(d-2)$-equilibrium measure on $\Sigma_{\boldsymbol{\gamma}}$ associated with~$Q$ and, by uniqueness, $\eta_{\Sigma_{\boldsymbol{\gamma}},Q} = \tau$ and $F_{\Sigma_{\boldsymbol{\gamma}},Q} = \mathcal{F}_s( \Sigma_{\boldsymbol{\gamma}} )$ with $C = \mathcal{F}_s( \Sigma_{\boldsymbol{\gamma}} ) / W_{d-2}( \mathbb{S}^d )$.

We show the variational inequality for $\PT{z} \in \mathbb{S}^d \setminus \Sigma_{\boldsymbol{\gamma}}$ and proceed in a similar fashion as in the proof of Theorem~\ref{thm:main}. 
For $i = 1, \ldots, m$ let
%
\begin{equation} \label{BalSeqS} 
\nu_i := {\sigma_d}_{\vert_{\Sigma_{i, \gamma_i}^c}}, \qquad \widehat{\nu_i} := \bal_{d-2}(\nu_i,\Sigma_{i,\gamma_i}), \qquad  \widehat{\delta_{\PT{a}_i}} := \bal_{d-2}( \delta_{\PT{a}_i}, \Sigma_{i,\gamma_i}),
\end{equation}
where $t_i$ is the projection of the boundary circle $\partial \Sigma_{i,\gamma_i}$ onto the $\PT{a}_i $-axis; recall that $2( 1 - t_i ) = \gamma_i^2$.
As the open spherical caps $\Sigma_{i,\gamma_i}^c$, $1 \leq i \leq m$, do not intersect for $i \neq j$, we have $\partial \Sigma_{\gamma_i} \subset \partial \Sigma_{\boldsymbol{\gamma}}$. 
Balayage in steps yields
\begin{align*}
\bal_{d-2}( \nu_{i}, \Sigma_{\boldsymbol{\gamma}} )
&= \bal_{d-2}(\nu_{i},\Sigma_{i,\gamma_i}) = \widehat{\nu_{i}} = W_{d-2}(\mathbb{S}^d ) \, \frac{1-t_i}{2} \left( 1 - t_i^2 \right)^{d/2-1} \beta_{i}, \\
\bal_{d-2}( \delta_{\PT{a}_i}, \Sigma_{\boldsymbol{\gamma}} ) 
&= \bal_{d-2}( \delta_{\PT{a}_i}, \Sigma_{i, \gamma_i} ) = \widehat{\delta_{\PT{a}_i}} = \frac{4}{\gamma_i^d} \, \frac{1-t_i}{2} \left( 1 - t_i^2 \right)^{d/2-1} \beta_{i},
\end{align*}
where the respective last step follow from \cite[Lemmas 33 and 36]{BraDraSaf2009} and it is crucial that $\widehat{\nu_{i}}$ and $\widehat{\delta_{\PT{a}_i}}$ are supported on $\partial \Sigma_{i,\gamma_i}$ and thus $\partial \Sigma_{\boldsymbol{\gamma}}$, so that
\begin{align}
\tau
&= C \, {\sigma_d}_{\vert_{\Sigma_{\boldsymbol{\gamma}}}} + C \, \sum_{i=1}^m \widehat{\nu}_{i} - \sum_{i=1}^m q_i \,  \widehat{\delta_{\PT{a}_i}} \label{SignedEqS} \\
&= C \, {\sigma_d}_{\vert_{\Sigma_{\boldsymbol{\gamma}}}} + \sum_{i=1}^m \left( C \, W_{d-2}(\mathbb{S}^d ) - \frac{4 q_i}{\gamma_i^d} \right) \frac{1-t_i}{2} \left( 1 - t_i^2 \right)^{d/2-1} \beta_{i}.
\end{align}
Observe, the signed measure $\tau$ has a negative component if and only if 
\begin{equation*}
C \, W_{d-2}(\mathbb{S}^d ) - \frac{4 q_i}{\gamma_i^d} < 0 \qquad \text{for at least one $i \in \{ 1, \dots, m \}$.}
\end{equation*}
Let $\PT{z} \in \mathbb{S}^d \setminus \Sigma_{\boldsymbol{\gamma}}$. Then $\PT{z} \in \Sigma_{i_0,\gamma_{i_0}}^c$ for some $i_0 \in \{1, \dots, m \}$ and $\PT{z} \in \Sigma_{i,\gamma_{i}}$ for all $i \neq i_0$. Hence, 
\begin{equation*}
U_{d-2}^{\tau}( \PT{z} ) + Q( \PT{z} ) = C \, W_{d-2}( \mathbb{S}^d ) +  C \left( U_{d-2}^{\widehat{\nu_{i_0}}}( \PT{z} ) - U_{d-2}^{\nu_{i_0}}( \PT{z} ) \right) - q_{i_0} \left( U_{d-2}^{\widehat{\delta_{\PT{a}_{i_0}}}}( \PT{z} ) - U_{d-2}^{\delta_{\PT{a}_{i_0}}}( \PT{z} ) \right).
\end{equation*}
Using \eqref{eq:balayage.partition}, from \cite[Lemmas 33]{BraDraSaf2009}
\begin{equation*}
U_{d-2}^{\widehat{\nu_{i_0}}}( \PT{z} ) - U_{d-2}^{\nu_{i_0}}( \PT{z} ) = W_{d-2}( \mathbb{S}^d ) \left( \frac{1+t_{i_0}}{1+\xi_{i_0}} \right)^{d/2-1} - W_{d-2}( \mathbb{S}^d ) < 0
\end{equation*}
and from \cite[Lemmas 36]{BraDraSaf2009},
\begin{equation*}
U_{d-2}^{\widehat{\delta_{\PT{a}_{i_0}}}}( \PT{z} ) - U_{d-2}^{\delta_{\PT{a}_{i_0}}}( \PT{z} ) = \frac{1}{\gamma_{i_0}^{d-2}} \left( \frac{1+t_{i_0}}{1+\xi_{i_0}} \right)^{d/2-1} - \frac{1}{\left( 2 ( 1 - \xi_{i_0} ) \right)^{d/2-1}} < 0;
\end{equation*}
hence
\begin{equation*}
\begin{split}
U_{d-2}^{\tau}( \PT{z} ) + Q( \PT{z} ) 
&= C \, W_{d-2}( \mathbb{S}^d ) \left( \frac{1+t_{i_0}}{1+\xi_{i_0}} \right)^{d/2-1} - \frac{q_{i_0}}{\gamma_{i_0}^{d-2}} \left( \frac{1+t_{i_0}}{1+\xi_{i_0}} \right)^{d/2-1} \\
&\phantom{=}+ \frac{q_{i_0}}{\left( 2 ( 1 - \xi_{i_0} ) \right)^{d/2-1}}.
\end{split}
\end{equation*}
Observe the similarity to \eqref{eq:weighted.outside.special}. Essentially the same argument as in the proof of Lemma~\ref{Lem} shows that
\begin{equation*}
U_{d-2}^{\tau}( \PT{z} ) + Q( \PT{z} ) > C \, W_{d-2}( \mathbb{S}^d ), \qquad \PT{z} \in \Sigma_{i_0,\gamma_{i_0}}
\end{equation*}
in the case when
\begin{equation} \label{eq:pos.cond.}
C \, W_{d-2}(\mathbb{S}^d ) - \frac{4 q_{i_0}}{\gamma_{i_0}^d} \leq 0, \qquad i = 1, \dots, m.
\end{equation}
It is not difficult to see that near $\partial\Sigma_{i_0,\gamma_{i_0}}$ the following asymptotics holds:
\begin{equation*}
\begin{split}
U_{d-2}^{\tau}( \PT{z} ) + Q( \PT{z} ) 
&= C \, W_{d-2}( \mathbb{S}^d ) + \left( \frac{d}{2} - 1 \right) \left( \frac{4 q_{i_0}}{\gamma_{i_0}^d} - C \, W_0( \mathbb{S}^d ) \right) \frac{\xi_{i_0} - t_{i_0}}{1 + t_{i_0}} \\
&\phantom{= C \, W_{d-2}( \mathbb{S}^d ) }+ \frac{1}{2} \left( \frac{d}{2} - 1 \right) \frac{d}{2} \left( \frac{4 q_{i_0}}{\gamma_{i_0}^d} \frac{2 t_{i_0}}{1+t_{i_0}} + C \, W_0( \mathbb{S}^d ) \right) \left( \frac{\xi_{i_0} - t_{i_0}}{1 + t_{i_0}} \right)^2 \\
&\phantom{= C \, W_{d-2}( \mathbb{S}^d ) }+ \mathcal{O}\Big( \left( \frac{\xi_{i_0} - t_{i_0}}{1 + t_{i_0}} \right)^3 \Big) \qquad \text{as $\xi_{i_0} \to t_{i_0}^+$;}
\end{split}
\end{equation*}
i.e., the weighted $(d-2)$-potential of $\tau$ will be negative sufficiently close to~$\partial\Sigma_{i_0,\gamma_{i_0}}$ if \eqref{eq:pos.cond.} does not hold. Hence, if the necessary conditions \eqref{eq:pos.cond.} 
are satisfied, then 
\begin{equation*}
U_{d-2}^{\tau}( \PT{z} ) + Q( \PT{z} ) > C \, W_{d-2}( \mathbb{S}^d ), \qquad \PT{z} \in \Sigma_{\boldsymbol{\gamma}}^c.
\end{equation*}

Suppose, the system 
\begin{align}
C \, W_{d-2}(\mathbb{S}^d ) &= \frac{4 q_{i}}{\gamma_{i}^d}, \qquad i = 1, \dots, m, \label{eq:gamma.i.cond.} \\
C \, \sigma_d( \Sigma_{\boldsymbol{\gamma}} ) &= 1, \label{EquationC}
\end{align}
subject to the geometric side conditions
\begin{equation}\label{GeomCond}
\Sigma_{i, \gamma_i } \bigcap \Sigma_{j, \gamma_j } = \emptyset, \qquad 1 \leq i < j \leq m,
\end{equation}
has a solution $( \boldsymbol{\gamma}, C )$ with $\boldsymbol{\gamma} = \boldsymbol{\gamma}( C ) \in (0,2)^m$ and $C > 0$, then $\eta_{\Sigma_{\boldsymbol{\gamma}},Q} = \tau = C \, {\sigma_d}_{\vert_{\Sigma_{\boldsymbol{\gamma}}}}$ with $F_{\Sigma_{\boldsymbol{\gamma}},Q} = C \, W_{d-2}(\mathbb{S}^d )$ satisfies the variational inequalities
\begin{equation}
\begin{split}\label{ExtSupport}
U_{d-2}^{\eta_{\Sigma_{\boldsymbol{\gamma}},Q}}( \PT{z} ) + Q( \PT{z} ) 
&= F_{\Sigma_{\boldsymbol{\gamma}},Q}, \qquad \PT{z} \in \Sigma_{\boldsymbol{\gamma}}, \\
U_{d-2}^{\eta_{\Sigma_{\boldsymbol{\gamma}},Q}}( \PT{z} ) + Q( \PT{z} ) 
&> F_{\Sigma_{\boldsymbol{\gamma}},Q}, \qquad \PT{z} \in \Sigma_{\boldsymbol{\gamma}}^c,
\end{split}
\end{equation}
and thus, by Proposition~\ref{prop:1}(d), $\mu_Q = \eta_{\Sigma_{\boldsymbol{\gamma}},Q} = C \, {\sigma_d}_{\vert_{\Sigma_{\boldsymbol{\gamma}}}}$ and $S_Q = \Sigma_{\boldsymbol{\gamma}}$.
Observe that, given a collection of pairwise different points $\PT{a}_1, \dots, \PT{a}_m \in \mathbb{S}^d$, for sufficiently small charges $q_1, \dots, q_m$, there always exists such a solution. In particular, this is the case if \eqref{EquationC} holds for $\gamma_i=[4q_i/W_{d-2}(\mathbb{S}^d)]^{1/d}$.
%
%
%
%
%
%
%

To determine the parameter $C$, denote $g(C):=C\sigma_d(\Sigma_{\boldsymbol{\gamma}} ) $, where 
\[\gamma_i :=\gamma_i(C)=\left[ \frac{4 q_i}{CW_{d-2}(\mathbb{S}^d)} \right]^{1/d}, \qquad i=1,\dots,m.\] 
As $\gamma_i = \gamma_i(C)$ are decreasing and continuous functions for all $i=1,\dots,m$, we derive that $\sigma_d(\Sigma_{\boldsymbol{\gamma}})$ is an increasing and continuous function of $C$ and so is $g(C)$. Also, note that $g(1) =\sigma_d (\Sigma_{\boldsymbol{\gamma}})<1$, and $\lim_{C\to \infty} g(C)=\infty$. Therefore, there exists a unique solution $C^*$ of the equation
\begin{equation*}
C\sigma_d \left( \bigcap_{i=1}^m \Sigma_{i, \gamma_i} \right) = 1,
\end{equation*}
where the $\gamma_i$'s are defined by \eqref{eq:gamma.i.cond.}. 
%

Finally, we invoke \cite[Corollary 13]{BraDraSaf2014} and \eqref{ExtSupport} to conclude that an optimal $(d-2)$-energy $N$-point configuration w.r.t. $Q$ is contained in $S_Q$.
\end{proof}

\section{Regions of Electrostatic Influence and Optimal $(d-2)$-Energy Points}
\label{Int}

In this section we consider what happens when the regions of electrostatic influence (see Remark~\ref{rmk:main} after Theorem~\ref{thm:main}) have intersecting interiors. We are going to utilize the techniques in the proofs of \cite[Theorem 14 and Corollary 15]{BraDraSaf2014} to show that the support of the $(d-2)$-equilibrium measure associated with the external field \eqref{DiscrField} satisfies $S_Q \subset \Sigma_{\boldsymbol{\epsilon}}$, and hence the optimal $(d-2)$-energy points stay away from $\Sigma_{\boldsymbol{\epsilon}}^c$. 
We are going to a prove our result for $s$ in the range $d-2\leq s<d$. 
%

Let $\PT{a}_1,\dots, \PT{a}_m \in \mathbb{S}^d$ be $m$ fixed points with associated positive charges $q_1,\dots, q_m$. We define for $d - 2 \leq s < d$ the external field
\begin{equation} \label{DiscrFieldS}
Q_s( \PT{x} ) :=\sum_{i=1}^m q_i \, k_{s}(\PT{a}_i , \PT{x}), \qquad \PT{x} \in \mathbb{S}^d.
\end{equation} 
We introduce the reduced charges
\begin{equation*}
\overline{q}_i := \frac{q_i}{1+q-q_i}, \qquad 1 \leq i \leq m.
\end{equation*}
Let $\overline{\Phi}_s(t_i)$ be the Mhaskar-Saff $\mathcal{F}_s$-functional associated with the external field $\overline{q}_i\, k_{s}(\PT{a}_i , \boldsymbol{\cdot})$ evaluated for the spherical cap $\Sigma_{i,\gamma_i}$ (cf. Section \ref{Rie}) where it is used that $t_i$ and $\gamma_i$ are related by $2(1-t_i)=\gamma_i^2$. Let $\overline{\gamma}_i$ denote the unique solution of the equation 
\begin{equation} \label{PhiBar} 
\overline{\Phi}_s(t_i) = \frac{2^{d-s} \overline{q}_i}{\gamma_i^d},\qquad 1 \leq i \leq m.
\end{equation}

\begin{theorem} \label{RegionInfluence} 
Let $d - 2 \leq s < d$, $d\geq 2$, and let $\overline{\boldsymbol{\gamma}} = (\overline{\gamma}_1, \dots, \overline{\gamma}_m )$ be the vector of solutions of \eqref{PhiBar}. Then the support $S_{Q_s}$ of the $s$-extremal measure $\mu_{Q_s}$ associated with the external field~$Q_s$ defined in \eqref{DiscrFieldS} is contained in the set
$\Sigma_{\overline{\boldsymbol{\gamma}}}=\bigcap_{i=1}^m \Sigma_{i,\overline{\gamma}_i}$. 
%
\linebreak
If $d=2$ and $s=0$, then $\overline{\gamma}_i=\epsilon_i$, $1 \leq i \leq m$, where $\epsilon_i$ is defined in \eqref{epsilon}.
%

Furthermore, no point of an optimal $N$-point configuration w.r.t. $Q_s$ lies in $\Sigma_{i,\overline{\gamma}_i}$, $1 \leq i \leq m$.
\end{theorem}

\begin{proof}
First, we consider the case $d-2<s<d$. 
Let $i$ be fixed. Since the external field \eqref{DiscrFieldS} has a singularity at $\PT{a}_i$, it is true that $S_{Q_s} \subset \Sigma_{i, \rho}$ for some $\rho >0$. Moreover, as noted after Definition~\ref{signed}, $S_{Q_s} \subset \supp( \eta_{\Sigma_{i, \gamma}, Q_s}^+)$ for all $\gamma$ such that $S_{Q_s} \subset \Sigma_{i,\gamma}$.
It is easy to see that the signed equilibrium on $\Sigma_{i,\gamma}$ associated with $Q_s$ is given by
\begin{equation}\label{signQBar} 
\eta_{\Sigma_{i,\gamma},Q_s} = \frac{1 + \sum_{j=1}^m q_j \, \| \widehat{\delta_{\PT{a}_j}} \|}{\| \widehat{\nu_{i}} \|} \, \widehat{\nu_{i}} - \sum_{j=1}^m q_j \,  \widehat{\delta_{\PT{a}_j}},
\end{equation}
where 
\begin{equation*}
\widehat{\nu_{i}} = \bal_s( \sigma_d, \Sigma_{i,\gamma}),\qquad \widehat{\delta_{\PT{a}_j}} = \bal_s( \delta_{\PT{a}_j}, \Sigma_{i, \gamma} ).
\end{equation*}
Observe, that if $\PT{a}_j \in \Sigma_{i,\gamma}$ then $\widehat{\delta_{\PT{a}_j}} = \delta_{\PT{a}_j}$. We will show that for all $\rho < \gamma < \overline{\gamma}_i$ the signed $s$-equilibrium measure in \eqref{signQBar} will be negative near the boundary $\partial \Sigma_{i,\gamma}$. \linebreak Indeed, with the convention that the inequality between two signed measures ${\nu_1 \leq \nu_2}$ means that ${\nu_2-\nu_1}$ is a non-negative measure, we have
\begin{align}
\eta_{\Sigma_{i,\gamma},Q_s} 
&\leq \frac{1 + \sum_{j=1}^m q_j \, \| \widehat{\delta_{\PT{a}_j}} \|}{\| \widehat{\nu_{i}} \|} \, \widehat{\nu_{i}} - q_i \, \widehat{\delta_{\PT{a}_i}} \notag \\ 
& \leq \left( 1 + q - q_i \right) \left( \frac{1+ \overline{q}_i \, \|\widehat{\delta_{\PT{a}_i}} \| }{\| \widehat{\nu_{i}} \|} \, \widehat{\nu_{i}} - \overline{q}_i \, \widehat{\delta_{\PT{a}_i}} \right) \notag \\
&= \left( 1 + q - q_i \right) \left[ \frac{\overline{\Phi}_s(t)}{W_s( \mathbb{S}^d)} \widehat{\nu_{i}} - \overline{q}_i \, \widehat{\delta_{\PT{a}_i}} \right], \label{SignedBar}
\end{align}
where $2( 1 - t ) = \gamma^2$.
The square-bracketed part is the signed equilibrium measure on~$\Sigma_{i,\gamma}$ associated with the external field $\overline{q}_i\, k_{s}(\PT{a}_i , \boldsymbol{\cdot})$ and has a negative component near the boundary $\partial \Sigma_{i,\gamma}$ if and only if $\overline{\Phi}_s(t) - \frac{2^{d-s} \overline{q}_i}{\gamma^d} < 0$ as noted after~\eqref{etabar}. 
%
%
%
This inequality holds whenever $\rho < \gamma < \overline{\gamma}_i$ and the inclusion relation $S_{Q_s} \subset \Sigma_{i, \gamma}$ for all $\rho < \gamma < \overline{\gamma}_i$ can now be easily deduced. As $i$ was arbitrarily fixed, we derive $S_{Q_s} \subset \Sigma_{\overline{\boldsymbol{\gamma}}}$. 
As an optimal $N$-point configuration w.r.t. $Q_s$ is confined to $S_{Q_s}$, no point of such a configuration lies in $\Sigma_{i,\overline{\gamma}_i}^c$, $1 \leq i \leq m$.


In order to obtain the result of the theorem for $d=2$ and $s=0$, we use that balayage under logarithmic interaction preserves mass. Hence
\begin{equation*}
\eta_{\Sigma_{i,\gamma},Q_0} = \overline{\Phi}_0( t ) \, \widehat{\nu_{i}} - \sum_{j=1}^m q_j \, \widehat{\delta_{\PT{a}_j}} \leq \overline{\Phi}_0( t ) \, \widehat{\nu_{i}} - \overline{q}_i \, \widehat{\delta_{\PT{a}_i}}, \qquad \overline{\Phi}_0( t ) := 1 + q 
\end{equation*}
and the characteristic equation $\overline{\Phi}_0( t ) = \frac{4 \overline{q}_i}{\overline{\gamma}^2}$ reduces to $\overline{\gamma}^2 = \frac{4 \overline{q}_i}{1 + \overline{q}_i}$.
As before, no point of an optimal $N$-point configuration w.r.t. $Q_0$ lies in $\Sigma_{i,\overline{\gamma}_i}^c$, $1 \leq i \leq m$. This completes the proof.
\end{proof}

%
%

\begin{figure}[ht]
\begin{center}
\includegraphics[scale=.27]{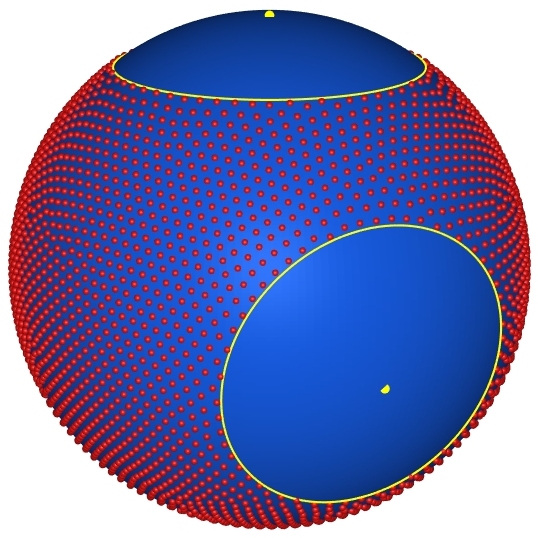} \qquad \includegraphics[scale=.27]{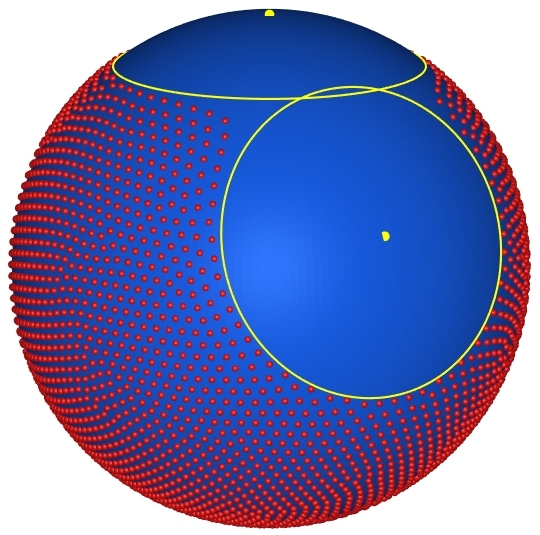} 
\end{center}
\caption{\label{fig3} Approximate Coulomb-optimal points for
$m = 2$, $N = 4000$, $q_1 = q_2 = \frac{1}{4}$, $\PT{a}_1 = ( 0, 0, 1 )$ and
$\PT{a}_2 = ( 0, \frac{\sqrt{91}}{10}, -\frac{3}{10} )$ or
$\PT{a}_2 = ( 0, \frac{\sqrt{91}}{10}, \frac{3}{10} )$}
\end{figure}

\begin{example} 
%
Observe, that if the charges $q_1, \dots, q_m$ are selected sufficiently small so that for all $i$ we have $\PT{a}_j \in \Sigma_{\gamma_i}$, then close to the boundary $\partial \Sigma_{\gamma_i}$ equality holds in \eqref{SignedBar}. So, the critical $\gamma_i$ can be determined by solving the equation
\[ \overline{\Phi}_s (t_i)-2^{d-s}\overline{q}_i /\gamma_i^d = 0,\]
where 
 \begin{equation}\label{PhiBarS} \overline{\Phi}_s (t_i) = W_s(\mathbb{S}^d )\frac{1+\overline{q}_i \|\widehat{\delta}_{t_i,s}\|}{\| \bal_s (\sigma_2,\Sigma_{\gamma_i})\|} .\end{equation}
Motivated by this, we consider the important case of Coulomb interaction potential, namely when $d=2$ and $s=1$. We find that (see \cite[Lemmas 29 and 30]{BraDraSaf2009})
\[ W_1 (\mathbb{S}^2 )=1,\quad \|\widehat{\delta}_{t_i,1}\| = \frac{\arcsin t_i}{\pi}+\frac{1}{2},\quad \| \bal_s (\sigma_2,\Sigma_{\gamma_i})\| = \frac{\sqrt{1-t_i^2} + \arcsin t_i}{\pi}+\frac{1}{2}.\]
Maximizing the Mhaskar-Saff $\mathcal{F}_1$-functional $\overline{\Phi}_1 (t)$ is equivalent to solving the equation
\[ \frac{\pi(1+\overline{q}_i/2)+\overline{q}_i \arcsin t_i}{\sqrt{1-t_i^2} + \arcsin t_i +\pi/2} =\frac{\overline{q}_i}{1-t_i}.\]
An equivalent equation in term of the geodesic radius $\alpha_i$ of the
cap $\Sigma_{\epsilon_i}^c$ of electrostatic influence,
so $t_i = \cos(\alpha_i)$ is
\[
  ({\bar q}_i+1)\pi \cos(\alpha) - {\bar q}_i \alpha \cos(\alpha) + {\bar q}_i \sin(\alpha)-\pi = 0.
\]
\end{example}

\begin{problem} The two images in Figure \ref{fig4} 
compare approximate log-optimal configurations with $4000$ and $8000$ points. 
The two yellow circles are the boundaries of $\Sigma_{1,\epsilon_1}$ and $\Sigma_{2,\epsilon_2}$. It is evident that optimal log-energy points stay away from the caps of electrostatic influence $\Sigma_{1,\epsilon_1}^c$ and $\Sigma_{2,\epsilon_2}^c$ of the two charges. In the limit, the log-optimal points approach the log-equilibrium support, which seems to be a smooth region excluding these caps of electrostatic influence. We conclude this section by posing as an open problem, the precise determination of the support in such a case. 
\end{problem}

\begin{figure}[ht]
\begin{center}
\includegraphics[scale=.28]{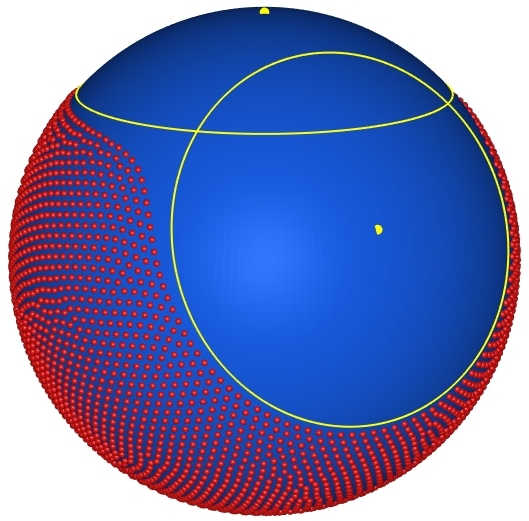} \quad \quad \includegraphics[scale=.28]{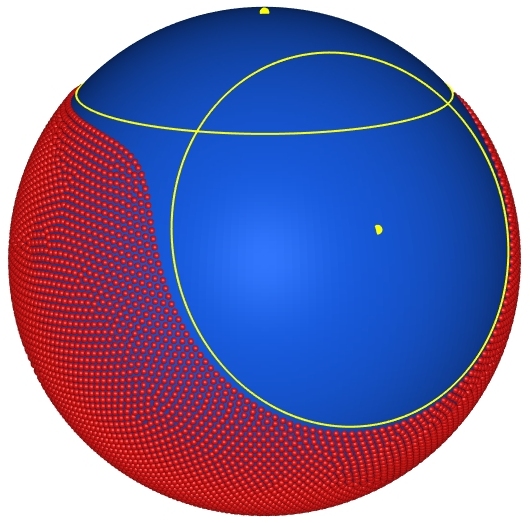} 
\end{center}
\caption{\label{fig4} Approximate log-optimal points for $m = 2$, $N = 4000$ (left) and $N = 8000$ (right), $q_1 = q_2 = \frac{1}{4}$, $\PT{a}_1 = ( 0, 0, 1 )$, $\PT{a}_2 = ( \frac{\sqrt{91}}{10}, 0, \frac{3}{10} )$}
\end{figure}

\begin{acknowledgement}
The research of Johann S. Brauchart was supported, in part, by the Austrian Science Fund FWF project F5510 (part of the Special Research Program (SFB) ``Quasi-Monte Carlo Methods: Theory and Applications'') and was also supported by the Meitner-Programm M2030 ``Self organization by local interaction'' funded by the Austrian Science Fund FWF. The research of Peter D. Dragnev was supported by the Simon's Foundation grant no. 282207. The research of Edward B. Saff was supported by U.S. National Science Foundation grant DMS-1516400. The research of Robert S. Womersley was supported by IPFW Scholar-in-Residence program. All the authors acknowledge the support of the Erwin Schr{\"o}dinger Institute in Vienna, where part of the work was carried out. This research includes computations using the Linux computational cluster Katana supported by the Faculty of Science, UNSW Sydney.
\end{acknowledgement}

%
\bibliographystyle{spmpsci}
\bibliography{REF}
%
%
%
%

\end{document}

%% file: macros.tex
\renewcommand{\email}[1]{\emailname: #1} 

\usepackage{mathptmx}       
\usepackage{helvet}         
\usepackage{courier}        

\usepackage{makeidx}         
\usepackage{graphicx}        
\usepackage[bottom]{footmisc}

\usepackage{latexsym}
\usepackage{amsmath}
\usepackage{amsfonts}
\usepackage{amssymb}
\usepackage{bm}

\usepackage{url}
\usepackage{algorithm}
\usepackage{algorithmic}
\usepackage[misc,geometry]{ifsym}

\renewenvironment{proof}{\noindent{\itshape Proof.}}{\smartqed\qed}

\spdefaulttheorem{assumption}{Assumption}{\upshape \bfseries}{\itshape}
\spdefaulttheorem{algo}{Algorithm}{\upshape \bfseries}{\itshape}

\DeclareSymbolFont{bbold}{U}{bbold}{m}{n}
\DeclareSymbolFontAlphabet{\mathbbold}{bbold}

